\documentclass[11pt]{article}
\usepackage[hmargin=1.1in,vmargin=1.3in]{geometry}

\title{
Sampling in a Quantum Population, and Applications \\[1ex]
}

\author{
Niek J. Bouman and Serge Fehr \\[1ex]
\small\it Centrum Wiskunde \& Informatica (CWI), Amsterdam, The Netherlands \\
\small\tt \{n.j.bouman,s.fehr\}@cwi.nl
}

\date{}

\usepackage{bbm}
\usepackage{amssymb}
\usepackage{amsmath}
\usepackage{amsfonts}
\usepackage{amsthm}
\usepackage{xspace}

\usepackage{comment}
\usepackage{booktabs}
\usepackage{graphicx}
\usepackage{ifpdf}
\ifpdf
\usepackage[usenames,dvipsnames]{color}
\else
\usepackage[usenames,dvips]{color}
\fi

\usepackage{verbatim}
\usepackage{ifthen}
\usepackage{hyperref}

\usepackage{float}
\floatstyle{ruled}
\newfloat{protocol}{thp}{lop}
\floatname{protocol}{Protocol}
\newcommand{\protoline}[2]{\item\emph{(#1)} #2\vspace{-0.5em}}

\definecolor{RemarkRed}{rgb}{0.8,0.0,0.0}

\ifpdf

\else

\fi

\renewcommand{\Pr}{\ensuremath{\mathrm{Pr}}}

\newcommand{\out}[2]{|#1\rangle \!\langle #2|\xspace}
\newcommand{\outs}[1]{|#1\rangle \!\langle #1|\xspace}
\newcommand{\inprod}[2]{\langle #1|#2 \rangle\xspace}
\renewcommand{\vec}[1]{\ensuremath{\boldsymbol{#1}}\xspace}

\newcommand{\B}{\ensuremath{{B_{t,s}^{\delta}}}\xspace}
\newcommand{\Brv}{\ensuremath{{B_{T,S}^{\delta}}}\xspace}
\newcommand{\tvec}[1]{\langle #1 |\xspace}
\newcommand{\kron}{\otimes\xspace}
\newcommand{\xor}{\oplus\xspace}

\newcommand{\ko}{\odot}
\newcommand{\kl}{\otimes}

\newcommand{\ho}{\texttt{+}}
\newcommand{\hl}{\texttt{-}}

\newcommand{\co}{0}
\newcommand{\cl}{1}

\newcommand{\ket}[1]{\ensuremath{| #1 \rangle}\xspace}         
\newcommand{\bra}[1]{\tvec{#1}}
\newcommand{\braket}[2]{\inprod{#1}{#2}}

\newcommand{\proj}[1]{\outs{#1}}
\newcommand{\tr}{\mathrm{tr}}
\newcommand{\I}{\mathbb{I}}
\newcommand{\dist}{\Delta}
\newcommand{\CNOT}{U_{\text{\sc cnot}}}

\newcommand{\set}[1]{\{#1\}}
\newcommand{\Set}[2]{\{#1:#2\}}
\newcommand{\setn}[1][n]{[n]}
\newcommand{\C}{\mathbb{C}}
\newcommand{\R}{\mathbb{R}}
\newcommand{\N}{\mathbb{N}}

\newcommand{\etal}{{\it et al.}\xspace}

\newcommand{\A}{{\cal A}}

\newcommand{\QOT}{{\bf\texttt{QOT}}\xspace}
\newcommand{\QOTx}{{\bf\texttt{QOT}}$^{\text{\bf\texttt{*}}}$\xspace}
\newcommand{\QKD}{{\bf\texttt{QKD}}\xspace}

\newcommand{\qt}{\ensuremath{\vec{q}\idx{t}}\xspace}
\newcommand{\qT}{\ensuremath{\vec{q}\idx{T}}\xspace}
\newcommand{\qbT}{\ensuremath{\vec{q}\idx{\bar{T}}}\xspace}
\newcommand{\func}{\ensuremath{f(t,\qt,s)}\xspace}

\newcommand{\refsec}[1]{Section~\ref{sec:#1}\xspace}
\newcommand{\refthm}[1]{Theorem~\ref{thm:#1}\xspace}
\newcommand{\reflem}[1]{Lemma~\ref{lem:#1}\xspace}
\newcommand{\refcor}[1]{Corollary~\ref{cor:#1}\xspace}
\newcommand{\refrem}[1]{Remark~\ref{rem:#1}\xspace}
\newcommand{\refex}[1]{Example~\ref{ex:#1}\xspace}
\newcommand{\refapp}[1]{Appendix~\ref{app:#1}\xspace}
\newcommand{\refdef}[1]{Definition~\ref{def:#1}\xspace}
\newcommand{\refprop}[1]{Proposition~\ref{prop:#1}\xspace}

\newcommand{\idx}[1]{_{#1}}
\newcommand{\weight}[1]{\ensuremath{\mathrm{wt}(#1)}\xspace}
\newcommand{\relwt}[1]{\ensuremath{\omega(#1)}\xspace}

\newcommand{\linspan}{\ensuremath{\mathrm{span}}}

\newcommand{\hmin}[1]{\ensuremath{\mathrm{H_{min}}\hspace{-1pt}\left(#1\right)}\xspace}

\newcommand{\hbin}{\ensuremath{\mathrm{h}}}
\newcommand{\minent}{min-entropy\xspace}

\newcommand{\id}{\ensuremath{\mathbb{I}}\xspace}

\newcommand{\cep}{\ensuremath{\varepsilon_{\mathrm{class}}^{\delta}}\xspace}
\newcommand{\Cep}{\ensuremath{\varepsilon_{\mathrm{class}}^{\delta}(\Psi)}\xspace}
\newcommand{\qep}{\ensuremath{\varepsilon_{\mathrm{quant}}^{\delta}}\xspace}
\newcommand{\Qep}{\ensuremath{\varepsilon_{\mathrm{quant}}^{\delta}(\Psi)}\xspace}
\renewcommand{\H}{\mathcal{H}}

\newcommand{\refeq}[1]{(\ref{eq:#1})\xspace}
\theoremstyle{plain} \newtheorem{thm}{Theorem}
\theoremstyle{plain} \newtheorem{lemma}{Lemma}
\theoremstyle{plain} \newtheorem{remark}{Remark}
\theoremstyle{plain} \newtheorem{definition}{Definition}
\theoremstyle{plain} \newtheorem{corollary}{Corollary}
\theoremstyle{plain} \newtheorem{prop}{Proposition}

\newtheoremstyle{example}{\topsep}{\topsep}%
     {}
     {}
     {\bfseries}
     {.}
     {\newline}
     {\thmname{#1}\thmnumber{ #2}\thmnote{ {\normalfont(#3)}}}

   \theoremstyle{example}
   \newtheorem{example}{Example} 

\newcommand{\niek}[1]{\noindent{\color{blue} \sf\small [Niek: #1]}}
\newcommand{\omt}[1]{}
\newcommand{\qkdsestrat}{
The subset $t \subset \setn \times \set{0,1}$ is chosen as $t = \set{(1,j_1),\ldots,(n,j_n)}$ where every $j_{k}$ is picked independently at random in $\set{0,1}$. In other words, $t$ selects one element from each pair $(q_{i0},q_{i1})$. 
Furthermore, the estimate for $\relwt{\vec{q}\idx{\bar{t}}}$ is computed from $\qt$ as $f(t,\qt,s) = \relwt{\vec{q}\idx{s}}$ where the seed $s$ is a random subset $s \subset t$ of size $k$. 
}
\newcommand{\qkdseerror}{
\newcommand{\qtT}{\ensuremath{\tilde{\vec{q}}\idx{T}}\xspace}
\newcommand{\qtbT}{\ensuremath{\tilde{\vec{q}}\idx{\bar{T}}}\xspace}

For $\A = \set{0,1}$, a bound on the error probability \cep is obtained as follows. 
Let $\vec{q}$ be arbitrary, indexed as discussed earlier. 
First, we show that $\relwt{\qbT}$ is likely to be close to $\relwt{\qT}$. 
For this, consider the pairs $(q_{i0},q_{i1})$ for which $q_{i0} \neq q_{i1}$. Let there be $\ell$ such pairs (where obviously $\ell\leq n$.) We denote the restrictions of \qT and \qbT to these indices $i$ with $q_{i0} \neq q_{i1}$ by \qtT and \qtbT, respectively. It is easy to see that $\weight{\qtT} +\weight{\qtbT} = \ell$. It follows that for any $\epsilon > 0$ we have  
\begin{align*}
\Pr\bigl[|  \relwt{\qbT}- \relwt{\qT} | &\geq \epsilon\bigr] 
= \Pr\bigl[\left| \weight{\qT} - \weight{\qbT} \right| \geq n\epsilon \bigr] \\[1ex]
&= \Pr\bigl[ \left| \weight{\qtT} - \weight{\qtbT} \right| \geq n\epsilon \bigr] 
=\Pr\bigl[\left|2\weight{\qtT} - \ell \right| \geq n\epsilon \bigr]\\
& \textstyle\leq 2\exp\left( -2 \left(\frac{n\epsilon}{2\ell}\right)^2 \ell\right) =  2\exp\left( -\frac{n\epsilon^2}{2}\cdot \frac{n}{\ell}\right) 
 \leq 2\exp\left( -\frac12 \epsilon^2 n\right) \, ,
\end{align*}
where the third equality follows from replacing $\weight{\qtbT} $ by $\ell - \weight{\qtT}$, and the first inequality follows from Hoeffding's inequality (as each entry of $\weight{\qtT}$ is $0$ with independent probability $\frac12$). 

Furthermore, for any $\gamma > 0$ we have the following relation involving $\vec{q}\idx{S}$:
\begin{align*}
\Pr\bigl[ \left| \relwt{\qT} - \relwt{\vec{q}\idx{S}} \right| \geq \gamma \bigr] \leq 2\exp\left(- 2 k \gamma^2 \right),
\end{align*}
which follows from directly applying Hoeffding's inequality.
Applying the union bound and letting $\delta =\epsilon+ \gamma$, we obtain
\begin{align*}
\label{eq:err}
\cep = \max_{\vec{q}} \Pr\bigl[\left| \relwt{\qbT} - \relwt{\vec{q}\idx{S}} \right| \geq \delta \bigr] & < 2 \min_{\epsilon \in (0, \delta)} \left[ \exp\left( -\textstyle\frac12 \epsilon^2 n \right) + \exp\left(- 2 k (\delta - \epsilon)^2 \right) \right] \\
&\textstyle\leq 4 \exp \left( -\frac{2 kn\delta ^2}{(2\sqrt{k}+\sqrt{n})^2} \right)
\leq 4 \exp \left( -\frac13 \delta ^2 k \right) \, ,
\end{align*}
where the last line follows from choosing $\epsilon$ such that the two exponents coincide, and from doing some simplifications while assuming $k \leq n/2$.
}

\newboolean{includeFifthEx} 
\setboolean{includeFifthEx}{true}

\begin{document}

\bibliographystyle{alpha}

\maketitle

\begin{abstract}
We propose a framework for analyzing classical sampling strategies for estimating the Hamming weight of a large string from a few sample positions, when applied to a multi-qubit quantum system instead. 
The framework shows how to interpret the result of such a strategy and how to define its accuracy when applied to a quantum system. 
Furthermore, we show how the accuracy of any strategy relates to its accuracy in its classical usage, which is well understood for the important examples. 

We show the usefulness of our framework by using it to obtain new and simple security proofs for the following quantum-cryptographic schemes: quantum oblivious-transfer from bit-commitment, and BB84 quantum-key-distribution. 
\paragraph{Keywords:} Random sampling, quantum key distribution, quantum oblivious transfer.
\end{abstract}

\section{Introduction}

Sampling allows to learn some information on a large population by merely looking at a comparably small number of individuals. For instance it is possible to predict the outcome of an election with very good accuracy by analyzing a relatively small subset of all the votes. 
In this work, we initiate the study of sampling in a {\em quantum} population, where we want to be able to learn information on a large quantum state by measuring only a small part. Specifically, we investigate the quantum-version of the following classical sampling problem (and of variants thereof). Given a bit-string  $\vec{q} = (q_1,\ldots,q_n) \in \set{0,1}^n$ of length $n$, the task is to estimate the Hamming weight of \vec{q} by sampling and looking at only a few positions within \vec{q}. This classical sampling problem is well understood. For instance the following particular {\em sampling strategy} works well: sample (with or without replacement) a linear number of positions uniformly at random, and compute an estimate for the Hamming weight of \vec{q} by scaling the Hamming weight of the sample accordingly; Hoeffding's bounds guarantee that the estimate is close to the real Hamming weight except with small probability. 
Such a sampling strategy in particular allows to {\em test} whether \vec{q} is close to the all-zero string $(0,\ldots,0)$ by looking only at a relatively small number of positions, where the test is accepted if and only if all the sample positions are zero, i.e., the estimated Hamming weight vanishes. 

In the quantum version of the above sampling problem, the string $\vec{q}$ is replaced by a $n$-qubit quantum system~$A$. It is obvious that a sampling strategy from the classical can be {\em applied} to the quantum setting as well: pick a sample of qubit positions within $A$, measure (in the computational basis) these sample positions, and compute the estimate as dictated by the sampling strategy from the observed values (i.e., typically, scale the Hamming weight of the measured sample appropriately). 
However, what is a-priori not clear, is how to formally {\em interpret} the computed estimate. 
In the special case of testing closeness to the all-zero string, one expects that if the measurement of a random sample only produces zeros then the initial state of $A$ must have been close to the all-zero state $\ket{\co}\cdots\ket{\co}$. But what is the right way to measure closeness here? For instance it must allow for states of the form $\ket{\vec{q}}$ where $\vec{q} \in \set{0,1}^n$ has small Hamming weight, but it must also allow for superpositions with arbitrary states that come with a very small amplitude. 
In the general case of a sampling strategy that, in its classical usage, aims at estimating the Hamming weight (rather that at testing closeness to the all-zero string), it is not even clear what the estimate actually estimates when the sampling strategy is applied to a $n$-qubit quantum system, since we cannot speak of the Hamming weight of a quantum state. 
Furthermore, how can we quantify in a meaningful way how accurate a sampling strategy is, and how hard is it to compute (good bounds on) the accuracy of different sampling strategies,  when applied to a quantum population? Finally, a last subtlety that is inherent to the quantum setting is that the execution of a sampling strategy actually changes the state of $A$ due to the measurements. 

In this work, we present a framework that answers the above questions and allows us to fully understand how a classical sampling strategy behaves when applied to a quantum population, i.e., to a $n$-qubit system or, more general, to $n$ copies of an arbitrary ``atomic'' system. 
Our framework incorporates the following. 
First, we specify an abstract property on the state of $A$ (after the measurements done by the sampling strategy), with the intended meaning that this is the property one should conclude from the outcome of the sampling strategy when applied to $A$. We also demonstrate that this property has useful consequences: specifically, that a suitable measurement will lead to a high-entropy outcome; this is handy in particular for quantum-cryptographic purposes. 
Then, we define a meaningful measure, sort of a ``quantum error probability'' (although technically speaking it is not a probability), that tells how reliable it is to conclude the specified property from the outcome of the sampling strategy. Finally, we show that for {\em any} sampling strategy, the quantum error probability of the strategy, as we define it, is bounded by the square-root of its classical error probability. This means that in order to understand how well a sampling strategy performs in the quantum setting, it suffices to analyze it in the classical setting, which is typically much simpler. Furthermore, for typical sampling strategies, like when picking the sample uniformly at random, there are well-known good bounds on the classical error probability. 

We demonstrate the usefulness of our framework by means of two applications. Our applications do not constitute actual new results, but they provide new and simple(r) proofs for known results, both in the area of quantum cryptography. We take this as strong indication for the usefulness of the framework, and that the framework is likely to prove valuable in other applications as well. 

The first application is to quantum oblivious transfer (QOT). It is well known that QOT is not possible from scratch; however, one can build a secure QOT scheme when given a bit-commitment (BC) primitive ``for~free''.%
\footnote{We use BC and OT as short-hands of the respective abstract primitives, bit commitment and oblivious transfer, and we write QBC and QOT for potential schemes implementing the respective primitives in the quantum setting. } 
Like QOT, also QBC is impossible from scratch; nevertheless, the implication from BC to QOT is interesting from a theoretical point of view, since the corresponding implication does not hold in the classical setting. 
The existence of a QOT scheme based on a BC was suggested by Bennett \etal in 1991~\cite{bennett1992};%
\footnote{At that time, QBC was thought to be possible, and thus the QOT scheme was claimed to be implementable from scratch. } 
however, no security proof was provided. Mayers and Salvail proved security of the QOT scheme against a restricted adversary that only performs {\em individual} measurements~\cite{mayers1994}, and finally, in 1995, Yao gave a security proof against a general adversary, which is allowed to do fully {\em coherent} measurements~\cite{yao1995}. However, from today's perspective, Yao's proof is still not fully satisfactory: it is very technical, without intuition and hard to follow, and it measures the adversary's information in terms of ``accessible information'', which has proven to be a too weak information measure.

Here, we show how our framework for analyzing sampling strategies in the quantum setting leads to a conceptually very simple and easy-to-understand security proof for QOT from BC. The proof essentially works as follows: When considering a purified version of the QOT scheme, the commit-and-open phase of the QOT scheme can be viewed as executing a specific sampling strategy. From the framework, it then follows that some crucial piece of information has high entropy from the adversary's point of view. The proof is then concluded by applying the privacy amplification theorem. 
In recent work of the second author \cite{damgaard2009}, it is shown that the same kind of analysis is not restricted to QOT but actually applies to a large class of two-party quantum-cryptographic schemes which are based on a commit-and-open phase.  

The second application we discuss is to quantum key-distribution (QKD). Also here, our framework allows for a simple and easy-to-understand security proof, namely for the BB84 QKD scheme.%
\footnote{Actually, we prove security for an entanglement-based version of BB84, which was first proposed by Ekert, and which implies security for the original BB84 scheme. } 
Similar to our proof for QOT, we can view the checking phase of the BB84 scheme as executing a specific sampling strategy (although here some additional non-trivial observation needs to be made).
From the framework, we can then conclude that the raw key has high entropy from the adversary's point of view, and again privacy amplification finishes the job. 

As for QOT, also QKD schemes initially came without security proofs, and proving QKD schemes rigorously secure turned out to be an extremely challenging and subtle task. Nowadays, though, the security of QKD schemes is better understood, and we know of various ways of proving, say, BB84 secure, ranging from Shor and Preskill's proof based on quantum error-correcting codes to Renner's approach using a quantum De Finetti theorem which allows to reduce security against general attacks to security against the much weaker class of so-called collective attacks. As such, our proof may safely be viewed as ``yet another BB84 QKD proof''. Nevertheless, when compared to other proofs, it has some nice features: It provides an explicit and easy-to-compute expression for the security of the scheme (in contrast to most proofs in the literature which merely provide an asymptotic analysis), it does not require any ``symmetrization of the qubits'' (e.g.\ by applying a random permutation) from the protocol, and it is technically not very involved (e.g.\ compared to the proofs involving Renner's quantum De Finetti theorem). Furthermore, it gives immediately a {\em direct} security proof, rather than a reduction to the security against collective attacks. 

\section{Notation, Terminology, and Some Tools}
\label{sec:notation}

\paragraph{Strings and Hamming Weight.}
Throughout the paper, $\A$ denotes some fixed finite alphabet with $0 \in \cal A$. It is safe to think of $\A$ as $\set{0,1}$, but our claims also hold for larger alphabets. For a string $\vec{q}=(q_1,\ldots,q_n)\in\A^n$ of arbitrary length $n~\geq~0$, the \emph{Hamming weight} of $\vec{q}$ is defined as the number of non-zero entries in $\vec{q}$: 
$\weight{\vec{q}} := \big|\Set{i \in \setn}{q_i \neq 0} \big|$, where we use $\setn$ as short hand for $\set{1,\ldots,n}$. 
We also use the notion of the \emph{relative} Hamming weight of $\vec{q}$, defined as 
$\relwt{\vec{q}} := \weight{\vec{q}}/n$. By convention, the relative Hamming weight of the empty string $\perp$ is set to $\relwt{\perp} := 0$.
For a string $\vec{q}\!=\!(q_1,\ldots,q_n)\!\in\!\A^n$ and a subset $J \subset \setn$, we write $\vec{q}\idx{J} := (q_i)_{i \in J}$ for the restriction of $\vec{q}$ to the positions $i \in J$. 

\paragraph{Random Variables and Hoeffding's Inequalities.}
Formally, a {\em random variable} is a function $X: \Omega \rightarrow \cal X$ with the sample space $\Omega$ of a probability space
$(\Omega,\Pr)$ as domain, and some arbitrary finite set $\cal X$ as range. The {\em distribution} of $X$, which we denote as $P_X$, is given by $P_X(x) = \Pr[X\!=\!x] = \Pr[\Set{\omega \in \Omega}{X(\omega)\!=\!x}]$. 
The joint distribution of two (or more) random variables $X$ and $Y$ is denoted by $P_{XY}$, i.e., $P_{XY}(x,y) = \Pr[X\!=\!x \wedge Y\!=\!y]$. 
Usually, we leave the probability space $(\Omega,\Pr)$ implicit, and understand random variables to be defined by their joint distribution, or by some ``experiment'' that uniquely determines their joint distribution. 
Random variables $X$ and $Y$ are {\em independent} if $P_{XY} = P_X P_Y$ (in the sense that $P_{XY}(x,y) = P_X(x) P_Y(y) \;\forall\, x\in {\cal X},y\in \cal Y$).

We will make extensive use of Hoeffding's inequalities for random sampling with and without replacement, as developed in~\cite{hoeffding1963}. The following theorem summarizes these inequalities, tailored to our needs.%
\footnote{Interestingly, the inequality with respect to random sampling {\em without} replacement does not seem to be very commonly known. }  

\begin{thm}[Hoeffding]
\label{thm:hoeffding}
Let $\vec{b} \in \set{0,1}^n$ be a bit string with relative Hamming weight $\mu = \relwt{\vec{b}}$. Let the random variables $X_1, X_2,\ldots, X_k$ be obtained by sampling $k$ random entries from $\vec{b}$ {\em with replacement}, i.e., the $X_i$'s are independent and $P_{X_i}(1) = \mu$. Furthermore, let the random variables $Y_1, Y_2,\ldots, Y_k$ be obtained by sampling $k$ random entries from $\vec{b}$ {\em without replacement}. Then, for any $\delta > 0$, the random variables $\bar{X}:=\frac{1}{k}\sum_i X_i$ and $\bar{Y}:=\frac{1}{k}\sum_i Y_i$ satisfy 
$$
\Pr\bigl[|\bar{Y} - \mu| \geq \delta\bigr] \leq \Pr\bigl[|\bar{X} - \mu| \geq \delta \bigr] \leq 2\exp(-2\delta^2 k) \, .
$$
For the case of sampling without replacement, a slightly sharper bound was found by Serfling \cite{serfling74}:
\[
\Pr\bigl[| \bar{Y} - \mu| \geq \delta\bigr] \leq 2\exp\bigl(-\textstyle\frac{2\delta^2kn}{n-k+1}\bigr).
\]
\end{thm}
\noindent In \cite{serfling74}, only a one-sided bound is given. We prove in \refapp{serfling} that this implies a two-sided bound.

\paragraph{Quantum Systems and States.}
We assume the reader to be familiar with the basic concepts of quantum information theory; we merely fix some terminology and notation here. 
A quantum system $A$ is associated with a complex Hilbert space, $\H = \C^d$, its {\em state space}. The {\em state} of $A$ is given, in the case of a {\em pure} state, by a norm-$1$ state vector $\ket{\varphi} \in \H$, respectively, in the case of a {\em mixed} state, by a trace-$1$ positive-semi-definite operator/matrix $\rho: \H \rightarrow \H$. In order to simplify language, we are sometimes a bit sloppy in distinguishing between a quantum system, its state, and the state vector or density matrix describing the state. 
By default, we write $\H_A$ for the state space of system $A$, and $\rho_A$ (respectively $\ket{\varphi_A}$ in case of a pure state) for the state of $A$. 

The state space of a {\em bipartite} quantum system $AB$, consisting of two (or more) subsystems, is given by $\H_{AB} = \H_A \otimes \H_B$. If the state of $AB$ is given by $\rho_{AB}$ then the state of subsystem $A$, when treated as a stand-alone system, is given by the {\em partial trace} $\rho_A = \tr_B(\rho_{AB})$, and correspondingly for $B$. 
{\em Measuring} a system $A$ in basis $\set{\ket{i}}_{i \in I}$, where $\set{\ket{i}}_{i \in I}$ is an orthonormal basis of $\H_A$, means applying the measurement described by the projectors $\set{\proj{i}}_{i \in I}$, such that outcome $i \in I$ is observed with probability $p_i = \tr(\proj{i} \rho_A)$ (respectively $p_i = |\braket{i}{\varphi_A}|^2$ in case of a pure state). 
If $A$ is a subsystem of a bipartite system $AB$, then it means applying the measurement described by the projectors $\set{\proj{i} \otimes \I_B}_{i \in I}$, where $\I_B$ is the identity operator on $\H_B$.

A {\em qubit} is a quantum system $A$ with state space~$\H_A = \C^2$. 
The {\em computational basis} $\set{\ket{\co},\ket{\cl}}$ (for a qubit) is given by $\ket{\co} = {1 \choose 0}$ and $\ket{\cl} = {0 \choose 1}$, and the {\em Hadamard basis} by $H\set{\ket{\co},\ket{\cl}} = \set{H\ket{\co},H\ket{\cl}}$, where $H$ denotes the 2-dimensional {\em Hadamard matrix} $H = \frac{1}{\sqrt2} \big(\begin{smallmatrix} 1 & \;\; 1 \\ 1 & -1 \end{smallmatrix}\big)$. 
The state space of an $n$-qubit system $A = A_1\cdots A_n$ is given by $\H_A = (\C^2)^{\otimes n} = \C^2 \otimes \cdots \otimes \C^2$. For $\vec{x} = (x_1,\ldots,x_n)$ and $\vec{\theta} = (\theta_1,\ldots,\theta_n)$ in $\set{0,1}^n$, we write $\ket{\vec{x}}$ for $\ket{\vec{x}} = \ket{x_1}\cdots\ket{x_n}$ and $H^{\vec{\theta}}$ for $H^{\vec{\theta}} = H^{\theta_1} \otimes \cdots \otimes H^{\theta_n}$, and thus $H^{\vec{\theta}}\ket{\vec{x}}$ for $H^{\vec{\theta}}\ket{\vec{x}} = H^{\theta_1}\ket{x_1}\cdots H^{\theta_n}\ket{x_n}$. Finally, we write $\set{\ket{\co},\ket{\cl}}^{\otimes n} = \Set{\ket{\vec{x}}}{\vec{x} \in \set{0,1}^n}$ for the computational basis on an $n$-qubit system, and  $H^{\vec{\theta}}\set{\ket{\co},\ket{\cl}}^{\otimes n} = \Set{H^{\vec{\theta}}\ket{\vec{x}}}{\vec{x} \in \set{0,1}^n} = H^{\theta_1}\set{\ket{\co},\ket{\cl}}\otimes\cdots\otimes H^{\theta_n}\set{\ket{\co},\ket{\cl}}$ for the basis that is made up of the computational basis on the subsystems $A_i$ with $\theta_i = 0$ and of the Hadamard basis on the subsystems $A_i$ with $\theta_i = 1$. 
In order to simplify notation, we will sometimes abuse terminology and speak of the basis $\vec{\theta}$ when we actually mean $H^{\vec{\theta}}\set{\ket{\co},\ket{\cl}}^{\otimes n}$. 

We measure closeness of two states $\rho$ and $\sigma$ by their {\em trace distance}: $\dist(\rho,\sigma) := \frac12 \tr|\rho-\sigma|$, where for any square matrix $M$, $|M|$ denotes the positive-semi-definite square-root of $M^\dagger M$. For {\em pure} states $\ket{\varphi}$ and $\ket{\psi}$, the trace distance of the corresponding density matrices coincides with $\dist(\proj{\varphi},\proj{\psi}) = \sqrt{1-|\braket{\varphi}{\psi}|^2}$. 
If the states of two systems $A$ and $B$ are $\epsilon$-close, i.e. $\dist(\rho_A,\rho_B) \leq \epsilon$, then $A$ and $B$ cannot be distinguished with advantage greater than $\epsilon$; in other words, $A$ behaves exactly like $B$, except with probability $\epsilon$. 

\paragraph{Classical and Hybrid Systems (and States).}
Subsystem $X$ of a bipartite quantum system $XE$ is called {\em classical}, if the state of $XE$ is given by a density matrix of the form
$$
\rho_{XE} = \sum_{x \in \cal X} P_X(x) \proj{x} \otimes \rho_{E}^x \, ,
$$
where $\cal X$ is a finite set of cardinality $|{\cal X}| = \dim(\H_X)$, $P_X:{\cal X} \rightarrow [0,1]$ is a probability distribution, $\set{\ket{x}}_{x \in \cal X}$ is some fixed orthonormal basis of $\H_X$, and $\rho_E^x$ is a density matrix on $\H_E$ for every \mbox{$x \in \cal X$}. Such a state, called {\em hybrid} or {\em cq-} (for {\em c}lassical-{\em q}uantum) state, can equivalently be understood as consisting of a {\em random variable} $X$ with distribution $P_X$, taking on values in $\cal X$, and a system $E$ that is in state $\rho_E^x$ exactly when $X$ takes on the value $x$. This formalism naturally extends to two (or more) classical systems $X$, $Y$ etc. 

If the state of $XE$ satisfies $\rho_{XE} = \rho_X \otimes \rho_E$, where $\rho_X = \tr_E(\rho_{XE}) = \sum_x P_X(x) \proj{x}$ and $\rho_E = \tr_X(\rho_{XE}) = \sum_x P_X(x) \rho_E^x$, then $X$ is {\em independent} of $E$, and thus no information on $X$ can be obtained from system~$E$. Moreover, if $\rho_{XE} = \frac{1}{|{\cal X}|} \I_X \otimes \rho_E$, where $\I_X$ denotes the identity on $\H_X$, then $X$ is {\em random-and-independent} of $E$. This is what is aimed for in quantum cryptography, when $X$ represents a classical cryptographic key and $E$ the adversary's potential quantum information on~$X$. 

It is not too hard to see that for two hybrid states $\rho_{XE}$ and $\rho_{XE'}$ with the same (distribution of) $X$, the trace distance between $\rho_{XE}$ and $\rho_{XE'}$ can be computed as $\dist(\rho_{XE},\rho_{XE'}) = \sum_x P_X(x) \dist(\rho_{E}^x,\rho_{E'}^x)$.

\paragraph{Min-Entropy and Privacy Amplification.}
We make use of Renner's notion of the {\em conditional min-entropy} $\hmin{\rho_{XE}|E}$ of a  system $X$ conditioned on another system $E$~\cite{Renner05}. 
Although the notion makes sense for arbitrary states, we restrict to hybrid states $\rho_{XE}$ with classical $X$. 
If the hybrid state $\rho_{XE}$ is clear from the context, we may write $\hmin{X|E}$ instead of $\hmin{\rho_{XE}|E}$.
The formal definition, given by $\hmin{\rho_{XE}|E}:= \sup_{\sigma_E}\max\Set{h \in \R}{2^{-h} \cdot \id_X \kron \sigma_E - \rho_{XE} \geq 0}$ where the supremum is over all density matrices $\sigma_E$ on $\H_E$, is not very relevant to us; we merely rely on some elementary properties. For instance, the {\em chain rule} guarantees that $\hmin{X|YE} \geq \hmin{XY|E} - \log(|{\cal Y}|) \geq \hmin{X|E} - \log(|{\cal Y}|)$ for classical $X$ and $Y$ with respective ranges $\cal X$ and $\cal Y$, where here and throughout the article $\log$ denotes the {\em binary} logarithm, whereas $\ln$ denotes the {\em natural} logarithm. Furthermore, it holds that if $E'$ is obtained from $E$ by measuring (part of) $E$, then $\hmin{X|E'} \geq \hmin{X|E}$.

Finally, we make use of Renner's privacy amplification theorem~\cite{RK05,Renner05}, as given below. 
Recall that a function $g:{\cal R} \times {\cal X} \rightarrow \set{0,1}^\ell$ is called a {\em universal} (hash) function, if for the random variable $R$, uniformly distributed over $\cal R$, and for any distinct $x,y \in \cal X$: $\Pr[g(R,x)\!=\!g(R,y)] \leq 2^{-\ell}$.

\begin{thm}[Privacy amplification]\label{thm:PA}
Let $\rho_{XE}$ be a hybrid state with classical $X$. Let $g:{\cal R} \times {\cal X} \to \set{0,1}^\ell$ be a universal hash function, and let $R$ be uniformly distributed over $\cal R$, independent of $X$ and~$E$. Then $K = g(R,X)$ satisfies
$$
\dist\bigl( \rho_{KRE},{\textstyle \frac{1}{|{\cal K}|}} \id_K \otimes \rho_{RE} \bigr) \leq \frac12 \cdot 2^{-\frac12(\hmin{X|E} - \ell)} \, . 
$$
\end{thm}

\noindent
Informally, Theorem~\ref{thm:PA} states that if $X$ contains sufficiently more than $\ell$ bits of entropy when given $E$, then $\ell$ nearly random-and-independent bits can be extracted from $X$.

\section{Sampling in a Classical Population}\label{sec:classic}

As a warm-up, and in order to study some useful examples and introduce some convenient notation, we start with the classical sampling problem, which is rather well-understood. 

\subsection{Sampling Strategies}

Let $\vec{q} = (q_1,\ldots,q_n) \in \mathcal{A}^n$ be a string of given length $n$. We consider the problem of estimating the relative Hamming weight $\relwt{\vec{q}}$ by only looking at a substring $\vec{q}\idx{t}$ of $\vec{q}$, for a small subset $t \subset \setn$.%
\footnote{More generally, we may consider the problem of estimating the Hamming {\em distance} of $\vec{q}$ to some arbitrary {\em reference string} $\vec{q_\circ}$; but this can obviously be done simply by estimating the Hamming weight of $\vec{q'} = \vec{q}-\vec{q_\circ}$.\label{fn:zero-string} }
Actually, we are interested in the equivalent problem of estimating the relative Hamming weight $\relwt{\vec{q}\idx{\bar{t}}}$ of the {\em remaining} string $\vec{q}\idx{\bar{t}}$, where $\bar{t}$ is the complement $\bar{t} = \setn \setminus t$ of $t$.%
\footnote{The reason for this, as will become clear later, is that in our applications, the sampled positions within $\vec{q}$ will be {\em discarded}, and thus we will be interested merely in the remaining positions. }
A canonical way to do so would be to sample a uniformly random subset (say, of a certain small size) of positions, and compute the relative Hamming weight of the sample as estimate. Very generally, we allow any strategy that picks a subset $t \subset \setn$ according to some probability distribution and computes the estimate for $\relwt{\vec{q}\idx{\bar{t}}}$ as some (possibly randomized) function of $t$ and $\qt$, i.e., as $\func$ for a {\em seed} $s$ that is sampled according to some probability distribution.
This motivates the following formal definition. 

\begin{definition}[Sampling strategy]\label{def:SE}
A {\em sampling strategy} $\Psi$ consists of a triple $(P_T,P_S,f)$, where $P_T$ is a distribution over the subsets of $\setn$, $P_S$ is a (independent) distribution over a finite set $\mathcal{S}$, and $f$ is a function
\begin{align*}
f: \{(t,v) : t \subset \setn, \vec{v} \in \mathcal{A}^{|t|}\} \times \mathcal{S}& \rightarrow \mathbb{R}. 
\end{align*}
\end{definition}
\noindent We stress that a sampling strategy $\Psi$, as defined here, specifies how to choose the sample subset as well as how to compute the estimate from the sample (thus a more appropriate but lengthy name would be a ``sample-and-estimate strategy''). 
\begin{remark}
\label{rem:container}
By definition, the choice of the seed $s$ is specified to be independent of~$t$, i.e., $P_{TS} = P_T P_S$. Sometimes, however, it is convenient to allow $s$ to depend on $t$. We can actually do so without contradicting \refdef{SE}. Namely, to comply with the independence requirement, we would simply choose a (typically huge) ``container'' seed that contains a seed for every possible choice of $t$, each one chosen with the corresponding distribution, and it is then part of $f$'s task, when given $t$, to select the seed that is actually needed out of the container seed.\footnote{Alternatively, we could simply drop the independence requirement in \refdef{SE}; however, we feel it is conceptually easier to think of the seed as being independently chosen. }
\end{remark}

\noindent 
A sampling strategy $\Psi$ can obviously also be used to {\em test} if $\vec{q}$ (or actually $\vec{q}\idx{\bar{t}}$) is close to the all-zero string $0\cdots0$: compute the estimate for $\relwt{\vec{q}\idx{\bar{t}}}$ as dictated by $\Psi$, and {\em accept} if the estimate vanishes and else {\em reject}. 

We briefly discuss \ifthenelse{\boolean{includeFifthEx}}{five}{four} example sampling strategies. The examples should illustrate the generality of the definition, and some of the examples will be used later on; however, the reader is free to skip (some of) them. We start with the canonical example mentioned in the beginning. 

\begin{example}[Random sampling {\em without} replacement]\label{ex:withoutreplacement}
In random sampling without replacement, $k$ {\em distinct} indices $i_1,\ldots,i_k$ within $\setn$ are chosen uniformly at random, where $k$ is some parameter, and the relative Hamming weight of $\vec{q}\idx{\set{i_1,\ldots,i_k}}$ is used as estimate for $\relwt{\vec{q}\idx{\bar{t}}}$.
Formally, this sampling strategy is given by $\Psi = (P_T,P_S,f)$ where \smash{$P_T(t) = 1/\binom{n}{k}$} 
if $|t| = k$ and else $P_T(t) = 0$, ${\cal S} = \set{\perp}$ and thus $P_S(\perp) = 1$, and $f(t,\qt,\perp) = f(t,\qt) = \relwt{\qt}$.\hfill$\diamond$ 
\end{example}
\noindent 
With the second example, we show that also sampling with replacement is captured by our definition. 

\begin{example}[Random sampling {\em with} replacement]\label{ex:withreplacement}
In random sampling with replacement, $k$ indices $i_1,\ldots,i_k$ are chosen independently uniformly at random within $\setn$, where $k$ is some parameter, and the relative Hamming weight of the string $(q_{i_1},\ldots,q_{i_k})$ is used as estimate for $\relwt{\vec{q}\idx{\bar{t}}}$. Note that here $i_{\ell}$ may coincide with $i_{\ell'}$ for $\ell \neq \ell'$, in which case $(q_{i_1},\ldots,q_{i_k})$ is not equal to $\vec{q}\idx{\set{i_1,\ldots,i_k}}$. 
To make this fit into \refdef{SE}, we set $t$ to be $\set{i_1,\ldots,i_k}$, and we let $f(t,\qt,s)$ be given by $\relwt{q_{j_1},\ldots,q_{j_k}}$, where $j_1,\ldots,j_k$ is determined by the seed $s$ among all possibilities with $\set{j_1,\ldots,j_k} = t$. 
It is cumbersome and of no importance to us to determine the correct distributions $P_T$ and $P_S$ for $t$ and $s$, respectively; it is sufficient to realize that random sampling with replacement {\em is} captured by \refdef{SE}. 
\hfill$\diamond$ 
\end{example}

\noindent Next, we sample by picking a uniformly random subset (without restricting its size). 
\begin{example}[Uniformly random subset sampling]
\label{ex:coinflip}
The sample set $t$ is chosen as a uniformly random subset of $\setn$, and the estimate is computed as the relative Hamming weight of the sample $\vec{q}\idx{t}$. Formally, $P_T(t) = 1/2^n$ for any $t \subseteq \setn$, and ${\cal S} = \set{\perp}$ and $f(t,\qt,\perp) = f(t,\qt) = \relwt{\qt}$.\hfill$\diamond$ 
\end{example}

\noindent 
As a fourth example, we consider a somewhat unnatural and in some sense non-optimal sampling strategy. This example, though, will be of use in our analysis of quantum oblivious transfer in \refsec{QOT}.

\newcommand{\exname}{Random sampling without replacement, using only part of the sample}
\begin{example}[\exname] 
\label{ex:QOT}
This example can be viewed as a composition of \refex{withoutreplacement} and \ref{ex:coinflip}. Namely, $t$ is chosen as a random subset of fixed size $k$, as in \refex{withoutreplacement}, so that $P_T(t) = 1/\binom{n}{k}$ for $t \subset \setn$ with $|t| = k$. But now, only part of the sample $\vec{q}\idx{t}$ is used to compute the estimate. Namely, the estimate is computed as 
\[
\func = \relwt{\vec{q}\idx{s}}.
\]
where the seed $s$ is chosen as a uniformly random subset $s$ of $t$; i.e., $P_S(s) = 1/2^t$ for any $s \subseteq t$. 
Recall from \refrem{container} that the choice of $s$ is allowed to depend on $t$. 
We would like to point out that when we use \refex{QOT} in Section~\ref{sec:QOT}, it is useful that the restriction to the subset $s$ is part of the evaluation of $f$, rather than part of the selection of the sample subset $t$. 
\hfill$\diamond$ 
\end{example}
\newcommand{\qkdexname}{Pairwise one-out-of-two sampling, using only part of the sample}
\ifthenelse{\boolean{includeFifthEx}}{
\noindent
In the fifth example we consider another somewhat unnatural sampling strategy, which though will be useful for the QKD proof in \refsec{QKD}.
\begin{example}[\qkdexname]\label{ex:qkdsestrat}
For this example, it is convenient to consider the index set from which the subset $t$ is chosen, to be of the form $\setn \times \set{0,1}$. Namely, we consider the string $\vec{q}~\in~\A^{2n}$ to be indexed by {\em pairs} of indices, $\vec{q} = (q_{ij})$, where $i \in \setn$ and $j \in \set{0,1}$; in other words, we consider $\vec{q}$ to consist of $n$ pairs $(q_{i0},q_{i1})$.
\qkdsestrat 
\hfill$\diamond$ 
\end{example}}{}
\newcommand{\biasexamplename}{Pairwise \emph{biased} one-out-of-two sampling, using only part of the sample\xspace}
\begin{example}[\biasexamplename]
\label{ex:bias}
In this example we consider a similar situation as in \refex{qkdsestrat}, except that we now construct $t$ by sampling every $j_k$ according to the \emph{Bernoulli} distribution $(p,1\!-\!p)$. Consequently, we compute the estimate for \relwt{\vec{q}\idx{\bar{t}}} slightly differently, but we will make this clear in \refapp{ex6}. \hfill$\diamond$ 
\end{example}

\subsection{The Error Probability}
After having introduced the general notion of a sampling strategy, we next want to define a measure that captures for a given sampling strategy how well it performs, i.e., with what probability the estimate, $f(t,\vec{q}\idx{t},s)$, is how close to the real value, $\relwt{\vec{q}\idx{\bar{t}}}$. For the definition, it will be convenient to introduce the following notation. For a given sampling strategy $\Psi = (P_T,P_S,f)$, consider arbitrary but fixed choices for the subset $t \subset \setn$ and the seed $s \in \cal S$ with $P_T(t) > 0$ and $P_S(s) > 0$. Furthermore, fix an arbitrary $\delta > 0$. Define $\B(\Psi) \subseteq \A^n$ as
\[
\B(\Psi) := \Set{\vec{b}\in \A^n}{ \left| \relwt{\vec{b}\idx{\bar{t}}} -  f(t,\vec{b}\idx{t},s) \right| < \delta } \, ,
\]
i.e., as the set of all strings $\vec{q}$ for which the estimate is $\delta$-close to the real value, assuming that subset $t$ and seed $s$ have been used. To simplify notation, if $\Psi$ is clear from the context, we simply write $\B$ instead of $\B(\Psi)$. 
By replacing the specific values $t$ and $s$ by the corresponding (independent) random variables $T$ and $S$, with distributions $P_T$ and $P_S$, respectively, we obtain the {\em random variable} $B_{T,S}^{\delta}$, whose range consists of subsets of $\A^n$. 
By means of this random variable, we now define the {\em error probability} of a sampling strategy as follows. 

\begin{definition}[Error probability]\label{def:class-err-prob}
The {\em (classical) error probability} of a sampling strategy $\Psi = (P_T,P_S,f)$ is defined as the following value, parametrized by $0 < \delta < 1$: 
\[
\Cep := \max_{\vec{q} \in \A^n} \Pr\Bigl[\vec{q} \notin B_{T,S}^{\delta}(\Psi) \Bigr] \, .
\]
\end{definition}

\noindent
By definition of the error probability, it is guaranteed that for any string $\vec{q} \in \A^n$, the estimated value is $\delta$-close to the real value except with probability at most $\Cep$. 
When used as a sampling strategy to test closeness to the all-zero string, $\Cep$ determines the probability of accepting even though $\vec{q}_{\bar{t}}$ is ``not close" to the all-zero string, in the sense that its relative Hamming weight exceeds $\delta$. Whenever $\Psi$ is clear from the context, we will write \cep instead of \Cep.

In Appendix~\ref{app:EPs}, we analyze the error probabilities for the sampling strategies considered in Examples \ref{ex:withoutreplacement} to \ifthenelse{\boolean{includeFifthEx}}{\ref{ex:qkdsestrat}}{\ref{ex:QOT}}, excluding Example~\ref{ex:withreplacement}, and we show them all to be exponentially small by applying Hoeffding's inequality in a suitable way. 

\section{Sampling in a {\em Quantum} Population}
\label{sec:quantsampling}

We now want to study the behavior of a sampling strategy when applied to a quantum population. More specifically, let $A = A_1\cdots A_n$ be an $n$-partite quantum system, where the state space of each system $A_i$ equals $\H_{A_i} = \C^d$ with $d = |\A|$, and let $\set{\ket{a}}_{a \in \A}$ be a fixed orthonormal basis of $\C^d$. 
We allow $A$ to be entangled with some additional system $E$ with arbitrary finite-dimensional state-space $\H_E$. 
We may assume the joint state of $AE$ to be pure, and as such be given by a state vector $\ket{\varphi_{AE}} \in \H_A \otimes \H_E$; if not, then it can be purified by increasing the dimension of $\H_E$. 

Similar to the classical sampling problem of testing closeness to the all-zero string, we can consider here the problem of testing if the state of $A$ is close to the all-zero {\em reference state} $\ket{\varphi_A^\circ} = \ket{\co}\cdots\ket{\co}$ by looking at, which here means {\em measuring}, only a few of the subsystems of $A$. 
More generally, we will be interested in the sampling problem of estimating the ``Hamming weight of the state of $A$", although it is not clear at the moment what this should mean. 
Actually, like in the classical case, we are interested in testing closeness to the all-zero state, respectively estimating the Hamming weight, of the {\em remaining subsystems} of $A$. 

It is obvious that a sampling strategy $\Psi = (P_T,P_S,f)$ can be applied in a straightforward way to the setting at hand: sample $t$ according to $P_T$, measure the subsystems $A_i$ with $i \in t$ in basis $\set{\ket{a}}_{a \in \A}$ to observe $\vec{q}\idx{t} \in \A^{|t|}$, and compute the estimate as $f(t,\vec{q}\idx{t},s)$ for $s$ chosen according to $P_S$ (respectively, for testing closeness to the all-zero state, accept or reject depending on the value of the estimate). However, it is a-priori {\em not} clear, how to interpret the outcome. Measuring a random subset of the subsystems of $A$ and observing 0 all the time indeed seems to suggest that the original state of $A$, and thus the remaining subsystems, must be in some sense close to the all-zero state; but what is the right way to formalize this? In the case of a general sampling strategy for estimating the (relative) Hamming weight, what does the estimate actually estimate? And, do all strategies that perform well in the classical setting also perform well in the quantum setting? 

We give in this section a rigorous analysis of sampling strategies when applied to a $n$-partite quantum system~$A$. Our analysis completely answers above concerns. Later in the paper, we demonstrate the usefulness of our analysis of sampling strategies for studying and analyzing quantum-cryptographic schemes. 

\subsection{Analyzing Sampling Strategies in the Quantum Setting}

We start by suggesting the property on the remaining subsystems of $A$ that one should expect to be able to conclude from the outcome of a sampling strategy. A somewhat natural approach is as follows. 

\begin{definition}
For system $AE$, and similarly for any subsystem of $A$, we say that the state $\ket{\varphi_{AE}}$ of $AE$ has {\em relative Hamming weight $\beta$ within $A$} if it is of the form $\ket{\varphi_{AE}} = \ket{\vec{b}} \ket{\varphi_E}$ with $\vec{b} \in \A^n$ and $\relwt{\vec{b}} = \beta$. 
\end{definition}

\noindent
Now, given the outcome $f(t,\vec{q}\idx{t},s)$ of a sampling strategy when applied to $A$, we want to be able to conclude that, up to a small error, the state of the remaining subsystem $A\idx{\bar{t}}E$ is a {\em superposition} of states with relative Hamming weight close to $f(t,\vec{q}\idx{t},s)$ within $A\idx{\bar{t}}$. To analyze this, we extend some of the notions introduced in the classical setting. Recall the definition of $\B$, consisting of all strings $\vec{b}\in \A^n$ with  $\left| \relwt{\vec{b}\idx{\bar{t}}} -  f(t,\vec{b}\idx{t},s) \right| < \delta$. 
By slightly abusing notation, we extend this notion to the quantum setting and write
$$
\linspan\bigl(\B\bigr) := \linspan\bigl( \Set{\ket{\vec{b}}}{\vec{b} \in \B}\bigr) = \linspan\bigl(\Set{\ket{\vec{b}}}{\left| \relwt{\vec{b}\idx{\bar{t}}} -  f(t,\vec{b}\idx{t},s) \right| < \delta}\bigr) \, .
$$
Note that if the state $\ket{\varphi_{AE}}$ of $AE$ happens to be in $\linspan(\B) \otimes \H_E$ for some $t$ and $s$, and if exactly these $t$ and $s$ are chosen when applying the sampling strategy to $A$, then {\em with certainty} the state of $A\idx{\bar{t}}E$ (after the measurement) is in a superposition of states with relative Hamming weight $\delta$-close to $f(t,\vec{q}\idx{t},s)$ within $A\idx{\bar{t}}$, regardless of the measurement outcome $\vec{q}\idx{t}$. 

Next, we want to extend the notion of error probability (Definition~\ref{def:class-err-prob}) to the quantum setting. 
The following approach turns out to be fruitful. We consider the {\em hybrid} system $TSAE$, consisting of the classical random variables $T$ and $S$ with distribution $P_{TS} = P_T P_S$, describing the choices of $t$ and $s$, respectively, and of the actual quantum systems $A$ and $E$. The state of $TSAE$ is given by 
\[
\rho_{TSAE} = \sum_{t,s} P_{TS}(t,s) \proj{t,s} \kron  \proj{\varphi_{AE}} \, .
\]
Note that $TS$ is independent of $AE$: $\rho_{TSAE} = \rho_{TS} \otimes \rho_{AE}$; indeed, in a sampling strategy $t$ and $s$ are chosen independently of the state of $AE$. We compare this {\em real} state of $TSAE$ with an {\em ideal} state which is of the form 
\begin{equation}\label{eq:IdealState}
\tilde{\rho}_{TSAE} =  \sum_{t,s} P_{TS}(t,s) \proj{t,s} \kron \out{\tilde{\varphi}^{ts}_{AE}}{\tilde{\varphi}^{ts}_{AE}}
\quad\text{with}\quad
\ket{\tilde{\varphi}^{ts}_{AE}} \in \linspan(\B) \!\otimes\! \H_E \;\;\forall\; t,s
\end{equation}
for some given $\delta > 0$. Thus, $T$ and $S$ have the same distribution as in the real state, but here we allow $AE$ to depend on $T$ and $S$, and for each particular choice $t$ and $s$ for $T$ and $S$, respectively, we require the state of $AE$ to be in $\linspan(\B) \otimes \H_E$. Thus, in an ``ideal world" where the state of the hybrid system $TSAE$ is given by $\tilde{\rho}_{TSAE}$, it holds {\em with certainty} that the state $\ket{\psi_{A\idx{\bar{t}}E}}$ of $A\idx{\bar{t}}E$, after having measured $A\idx{t}$ and having observed $\vec{q}\idx{t}$, is in a superposition of states with relative Hamming weight $\delta$-close to $\beta:= f(t,\vec{q}\idx{t},s)$ within $A\idx{\bar{t}}$. We now define the quantum error probability of a sampling strategy by looking at how far away the closest ideal state $\tilde{\rho}_{TSAE}$ is from the real state $\rho_{TSAE}$. 

\begin{definition}[Quantum error probability]
The {\em quantum error probability} of a sampling strategy $\Psi = (P_T,P_S,f)$ is defined as the following value, parametrized by $0 < \delta < 1$: 
\[
\Qep = \max_{\H_E} \max_{\ket{\varphi_{AE}}}\min_{\tilde{\rho}_{TSAE}} \dist(\rho_{TSAE},\tilde{\rho}_{TSAE}) \, ,
\]
where the first $\max$ is over all finite-dimensional state spaces $\H_E$, the second $\max$ is over all state vectors $\ket{\varphi_{AE}} \in \H_A \otimes \H_E$, and the $\min$ is over all ideal states $\tilde{\rho}_{TSAE}$ as in \eqref{eq:IdealState}.%
\footnote{It is not too hard to see, in particular after having gained some more insight via the proof of Theorem~\ref{thm:quanterr} below, that these $\min$ and $\max$ exist. }
\end{definition}

\noindent 
As with \B and \cep, we simply write \qep when $\Psi$ is clear from the context.
We stress the meaningfulness of the definition: it guarantees that on average over the choice of $t$ and $s$, the state of $A\idx{\bar{t}}E$ is $\qep$-close to a superposition of states with Hamming weight $\delta$-close to $\func$ within $A\idx{\bar{t}}$, and as such it {\em behaves} like a superposition of such states, except with probability $\qep$. We will argue below and demonstrate in the subsequent sections that being (close to) a superposition of states with given approximate (relative) Hamming weight has some useful consequences. 

\begin{remark}\label{rem:allzero} 
Similarly to footnote~\ref{fn:zero-string}, also here the results of the section immediately generalize from the all-zero reference state $\ket{\co}\cdots\ket{\co}$ to an arbitrary reference state $\ket{\varphi_A^\circ}$ of the form $\ket{\varphi_A^\circ} = U_1 \ket{\co} \otimes \cdots \otimes U_n \ket{\co}$ for unitary operators $U_i$ acting on $\C^d$. Indeed, the generalization follows simply by a suitable change of basis, defined by the $U_i$'s. 
Or, in the special case where $\A = \set{0,1}$ and 
$$
\ket{\varphi_A^\circ} = H^{\vec{\hat{\theta}}}\ket{\vec{\hat{x}}} = H^{\hat{\theta}_1} \ket{\hat{x}_1} \otimes \cdots \otimes H^{\hat{\theta}_n} \ket{\hat{x}_n}
$$
for a fixed reference basis $\vec{\hat{\theta}} \in \set{0,1}^n$ and a fixed reference string $\vec{\hat{x}} \in \set{0,1}^n$, we can, alternatively, replace in the definitions and results the computational by the Hadamard basis whenever $\hat{\theta}_i = 1$, and speak of the (relative) Hamming distance to $\vec{\hat{x}}$ rather than of the (relative) Hamming weight. 
\end{remark}

\subsection{The Quantum vs. the Classical Error Probability}

It remains to discuss how difficult it is to actually {\em compute} the quantum error probability for given sampling strategies, and how the {\em quantum} error probability $\qep$ relates to the corresponding {\em classical} error probability $\cep$. To this end, we show the following simple relationship between $\qep$ and $\cep$. 

\begin{thm}
\label{thm:quanterr}
For any sampling strategy $\Psi$ and for any $\delta > 0$: 
\[
\Qep \leq \sqrt{\Cep}.
\]
\end{thm}

\noindent
As a consequence of this theorem, it suffices to analyze a sampling strategy in the classical setting, which is much easier, in order to understand how it behaves in the quantum setting. In particular, sampling strategies that are known to behave well in the classical setting, like examples~\ref{ex:withoutreplacement} to~\ifthenelse{\boolean{includeFifthEx}}{\ref{ex:qkdsestrat}}{\ref{ex:QOT}}, are also automatically guaranteed to behave well in the quantum setting. We will use this in the application sections. 

Our bound on $\qep$ is in general tight. Indeed, in \refapp{tight} we show tightness for an explicit class of sampling strategies, which e.g.\ includes \refex{withoutreplacement} and \refex{qkdsestrat}. Here, we just mention the tightness result. 
\begin{prop}
\label{prop:natural}
There exist natural sampling strategies for which the inequality in \refthm{quanterr} is an equality.
\end{prop}
\begin{proof}[Proof of \refthm{quanterr}] 
We need to show that for any $\ket{\varphi_{AE}} \in \H_A \otimes \H_E$, with arbitrary $\H_E$, there exists a suitable ideal state $\tilde{\rho}_{TSAE}$ with $\dist(\rho_{TSAE},\tilde{\rho}_{TSAE}) \leq (\cep)^{1/2}$. 
We construct $\tilde{\rho}_{TSAE}$ as in \refeq{IdealState}, where the $\ket{\tilde{\varphi}^{ts}_{AE}}$'s are defined by the following decomposition.
\[
\ket{\varphi_{AE}} = \inprod{\tilde{\varphi}^{ts}_{AE}}{\varphi_{AE}} \ket{\tilde{\varphi}^{ts}_{AE}} + \inprod{\tilde{\varphi}^{ts\perp}_{AE}}{\varphi_{AE}} \ket{\tilde{\varphi}^{ts\perp}_{AE}},
\]
with $\ket{\tilde{\varphi}^{ts}_{AE}} \in \linspan(\B) \otimes \H_E$, $\ket{\tilde{\varphi}^{ts\perp}_{AE}}\in \linspan(\B)^\perp \otimes \H_E$ and
$|\inprod{\tilde{\varphi}^{ts}_{AE}}{\varphi_{AE}}|^2+|\inprod{\tilde{\varphi}^{ts\perp}_{AE}} {\varphi_{AE}} |^2=1$.
In other words, $\ket{\tilde{\varphi}^{ts}_{AE}}$ is obtained as the re-normalized projection of $\ket{\varphi_{AE}}$ into $\linspan(\B) \otimes \H_E$. Note that $|\inprod{\tilde{\varphi}^{ts\perp}_{AE}}{\varphi_{AE}}|^2$ equals the probability $\Pr\bigl[\vec{Q} \!\notin\! \B\bigr]$, where the random variable $\vec{Q}$ is obtained by measuring subsystem $A$ of $\ket{\varphi_{AE}}$ in basis $\set{\ket{a}}_{a \in \cal A}^{\otimes n}$. Furthermore, 
$$
\sum_{t,s} P_{TS}(t,s) \, |\inprod{\tilde{\varphi}^{ts\perp}_{AE}}{\varphi_{AE}}|^2 =
\sum_{t,s} P_{TS}(t,s) \, \Pr\bigl[\vec{Q} \!\notin\! \B\bigr] =
\Pr\bigl[\vec{Q} \!\notin\! B_{T,S}^{\delta}\bigr]
= \sum_{\vec{q}} P_{\vec{Q}}(\vec{q})\, \Pr\bigl[\vec{q} \!\notin\! B_{T,S}^{\delta}\bigr] ,
$$ 
where by definition of $\cep$, the latter is upper bounded by $\cep$. From elementary properties of the trace distance, and using Jensen's inequality, we can now conclude that
\begin{align*}
\dist\bigl(\rho_{TSAE}&,\tilde{\rho}_{TSAE}\bigr) = \sum_{t,s}  P_{TS}(t,s) \dist\bigl(\proj{\varphi_{AE}}, \proj{\tilde{\varphi}^{ts}_{AE}}\bigr)  
= \sum_{t,s} P_{TS}(t,s) \sqrt{1 - |\inprod{\tilde{\varphi}^{ts}_{AE}}{\varphi_{AE}}|^2} \\
&= \sum_{t,s} P_{TS}(t,s) |\inprod{\tilde{\varphi}^{ts\perp}_{AE}}{\varphi_{AE}}|  
\leq \sqrt{\sum_{t,s} P_{TS}(t,s) |\inprod{\tilde{\varphi}^{ts\perp}_{AE}}{\varphi_{AE}}|^2} \leq \sqrt{\cep},
\end{align*}
which was to be shown. 
\end{proof}
As a side remark, we point out that the particular ideal state $\tilde{\rho}_{TSAE}$ constructed in the proof minimizes the distance to $\rho_{TSAE}$; this follows from the so-called Hilbert projection theorem. 

\subsection{Superpositions with a Small Number of Terms}

We give here some argument why being (close to) a superposition of states with a given approximate Hamming weight may be a useful property in the analyses of quantum-cryptographic schemes. 
For simplicity, and since this will be the case in our applications, we now restrict to the binary case where $\A = \set{0,1}$. 
Our argument is based on the following lemma, which follows immediately from Lemma~3.1.13 in~\cite{Renner05}; for completeness, we give a direct proof of Lemma~\ref{lem:minentsup} in Appendix~\ref{app:minentsup}. 
Informally, it states that measuring (part of) a {\em superposition} of a small number of orthogonal states produces a similar amount of uncertainty as when measuring the {\em mixture} of these orthogonal states. 

\begin{lemma}
\label{lem:minentsup}
Let $A$ and $E$ be arbitrary quantum systems, let $\set{\ket{i}}_{i \in I}$ and $\{\ket{w}\}_{w\in \mathcal{W}}$ be orthonormal bases of $\H_A$, and let $\ket{\varphi_{AE}}$ and $\rho_{AE}^\mathrm{mix}$ be of the form
\[
\ket{\varphi_{AE}} = \sum_{i \in J} \alpha_{i} \ket{i}\ket{\varphi_E^{i}} \in \H_A \otimes \H_E
\qquad\text{and}\qquad
\rho_{AE}^\mathrm{mix} = \sum_{i \in J} |\alpha_{i}|^2 \outs{i}\kron \outs{\varphi^{i}_E} 
\]
for some subset $J \subseteq I$. 
Furthermore, let $\rho_{WE}$ and $\rho_{WE}^\mathrm{mix}$ describe the hybrid systems obtained by measuring subsystem $A$ of \ket{\varphi_{AE}} and $\rho_{AE}^\mathrm{mix}$, respectively, in basis $\{\ket{w}\}_{w\in \mathcal{W}}$ to observe outcome $W$. Then, 
\[
\hmin{\rho_{WE}|E} \geq \hmin{\rho_{WE}^\mathrm{mix} | E} - \log |J| \, .
\]
\end{lemma}

\noindent
We apply Lemma~\ref{lem:minentsup} to an $n$-qubit system $A$ where $\ket{\varphi_{AE}}$ is a superposition of states with relative Hamming weight $\delta$-close to $\beta$ within $A$:%
\footnote{System $A$ considered here corresponds to the subsystem $A\idx{\bar{t}}$ in the previous section, after having measured $A\idx{t}$ of the ideal state.}
$$
\ket{\varphi_{AE}} = \sum_{\vec{b} \in \set{0,1}^n \atop |\relwt{\vec{b}} - \beta| \leq \delta} \ket{\vec{b}}\ket{\varphi_E^{\vec{b}}} \, .
$$
It is well known that $\big|\Set{\vec{b} \in \set{0,1}^n}{|\relwt{\vec{b}} - \beta| \leq \delta}\big| \leq\big|\Set{\vec{b} \in \set{0,1}^n}{\relwt{\vec{b}} \leq \beta + \delta}\big| \leq 2^{\hbin(\beta+\delta)n}$ for $\beta + \delta \leq \frac12$, 
 where the function $\hbin: [0,1] \rightarrow [0,1]$ is the {\em binary entropy function}, defined as $\hbin(p) = -\bigl( p \log(p) + (1-p) \log(1-p) \bigr)$ for $0 < p < 1$ and as $0$ for $p = 0$ or $1$.%
\footnote{There exists a corresponding upper bound for the cardinality of a $q$-ary Hamming ball (with arbitrary $q$), expressed in terms of the so-called $q$-ary entropy function; we do not elaborate on this here, since we now focus on the binary case. }

Since measuring qubits within a state $\ket{\vec{b}}$ in the {\em Hadamard} basis produces uniformly random bits, we can conclude the following. 
\begin{corollary}\label{cor:minentsup}
Let $A$ be an $n$-qubit system, let the state $\ket{\varphi_{AE}}$ of $AE$ be a superposition of states with relative Hamming weight $\delta$-close to $\beta$ within $A$, where $\delta + \beta \leq \frac12$, and let the random variable $\vec{X}$ be obtained by measuring $A$ in basis $H^{\vec{\theta}} \set{\ket{\co},\ket{\cl}}^{\otimes n}$ for $\vec{\theta} \in \set{0,1}^n$. Then 
$$
\hmin{\vec{X}|E} \geq \weight{\vec{\theta}} - \hbin(\beta+\delta) n \, .
$$
\end{corollary}

\noindent
Consider now the following quantum-cryptographic setting. Bob prepares and hands over to Alice an $n$-qubit quantum system $A$, which ought to be in state $\ket{\varphi_A^\circ} = \ket{\co}\cdots\ket{\co}$. However, since Bob might be dishonest, the state of $A$ could be anything, even entangled with some system $E$ controlled by Bob. Our results now imply the following: Alice can apply a suitable sampling strategy to convince herself that the joint state of the remaining subsystem of $A$ and of $E$ is (close to) a superposition of states with bounded relative Hamming weight. From Corollary~\ref{cor:minentsup}, we can then conclude that with respect to the min-entropy of the measurement outcome, the state of $A$ behaves similarly to the case where Bob honestly prepares $A$ to be in state $\ket{\varphi_A^\circ}$. 
By Remark~\ref{rem:allzero}, i.e., by doing a suitable change of basis, the same holds if \smash{$\ket{\varphi_A^\circ} = H^{\vec{\hat{\theta}}}\ket{\vec{\hat{x}}}$} for arbitrary fixed $\vec{\hat{\theta}},\vec{\hat{x}} \in \set{0,1}^n$, where $\weight{\vec{\theta}}$ is replaced by the Hamming distance between $\vec{\theta}$ and $\vec{\hat{\theta}}$. 
We will make use of this in the applications in the upcoming sections. 

\section{Application I: Quantum Oblivious Transfer (QOT)}\label{sec:QOT}

\subsection{The Bennett \textit{et al.} QOT Scheme}

In a (one-out-of-two) {\em oblivious transfer}, OT for short, Alice sends two messages, $\vec{m}_0,\vec{m}_1 \in \set{0,1}^\ell$ to Bob. Bob may choose to receive one of the two message, $\vec{m}_c$. The security requirements demand that Bob learns no information on the other message, $\vec{m}_{1-c}$, while at the same time Alice remains ignorant about Bob's choice bit~$c$. 

Back in 1992, Bennett \etal proposed a quantum scheme for OT, i.e., a QOT scheme~\cite{bennett1992}. The scheme makes use of a {\em bit commitment} (BC), which at that point in time was believed to be implementable with unconditional security by a quantum scheme. Bennett \etal, however, merely claimed security of their scheme without providing any proof. In 1994, Mayers and Salvail proved the QOT scheme secure against a limited class of attacks~\cite{mayers1994}, and, subsequently, Yao presented a full security proof without limiting the adversary's capabilities~\cite{yao1995}. However, Yao's proof is lengthy and very technical, and thus hard to understand. Furthermore, security is phrased and proven in terms of {\em accessible information}, of which we now know that it is a too weak information measure to guarantee security as required. 

Here we show how our sampling-strategy framework naturally leads to a new security proof for Bennett \etal's QOT scheme. The new proof is simple and conceptually easy-to-understand, and security is expressed and proven by means of a security definition that is currently accepted to be ``the right one''. Furthermore, it allows for an explicit bound on the imperfection of the scheme for any set of parameters (number of transmitted qubits, length of messages etc.), rather than merely providing an asymptotic security claim. 
Nowadays, we of course know that BC (as well as QOT) cannot be implemented with unconditional security by means of a quantum scheme: QBC is impossible~\cite{Mayers97,lo1997}. As such QOT cannot be instantiated from scratch.  
Nevertheless, the existence of a QOT scheme based on a (hypothetical) BC is still an interesting result, since in the non-quantum world, a BC alone does {\em not} allow to implement OT. 

Below, we describe Bennett \etal's QOT scheme (with some minor modifications), which we denote as \QOT. 
Actually, \QOT corresponds to the {\em randomized} OT used within Bennett \etal's QOT scheme, where the messages $\vec{m}_0$ and $\vec{m}_1$, called $\vec{k}_0$ and $\vec{k}_1$ in \QOT, are not {\em input} by Alice (her input is empty:~$\perp$) but randomly produced during the course of the scheme and then {\em output} to Alice. The desired non-randomized OT is then obtained simply by one-time-pad encrypting Alice's input messages $\vec{m}_0$ and $\vec{m}_1$ with the keys $\vec{k}_0$ and $\vec{k}_1$, respectively. Security of the non-randomized OT follows immediately from the security of the randomized OT by the properties of the one-time-pad.

\QOT is parametrized by parameters $n,k,\ell \in \N$, where $n$ is the number of qubits communicated, $\ell$ the bit-length of the messages/keys $\vec{k}_0, \vec{k}_1$, and $k$ is the size of the ``test set'' $t$, which we require to be at most $n/2$. \QOT makes use of a universal hash function $g:{\cal R} \times \set{0,1}^n \rightarrow \set{0,1}^\ell$. For $\vec{x}' \in \set{0,1}^{n'}$ with $n' < n$, we define $g(r,\vec{x}')$ as $g(r,\vec{x})$ where $\vec{x} \in \set{0,1}^n$ is obtained from $\vec{x}'$ by padding it with sufficiently many $0$'s. 
Furthermore, the scheme makes use of a BC, which we assume to be an ideal BC functionality. Alternatively, at the cost of losing unconditional security against dishonest Alice, we may use a BC implementation that is perfectly binding and computationally hiding.%
\footnote{Note that we do not claim any kind of composability for this computational setting. In case of a perfectly hiding and computationally binding BC scheme, our techniques do not apply directly. A specific variant of the latter case (in which the BC is required to have some additional properties) is handled in~\cite{damgaard2009}.}
Finally, for simplicity, we assume a {\em noise-free} quantum channel. 
For the more realistic setting of noisy quantum communication, an error-correcting code can be applied in a similar fashion as in the original scheme; this will not significantly affect our proof. 
In the upcoming protocol%
\footnote{A {\em protocol} is an interactive algorithm between two (or in general more) entities, whereas a {\em scheme} in general may consist of several protocols (like for BC); since the cryptographic tasks considered in this article (QOT and QKD) ask for a single protocol, we use the terms {\em protocol} and {\em scheme} interchangeably.} 
descriptions, we make use of our convention to speak about a basis $\vec{\theta}$ (or $\vec{\hat{\theta}}\,$) in $\set{0,1}^n$ when we actually mean $H^{\vec{\theta}} \{ \ket{\co},\ket{\cl} \}^{\otimes n}$ (respectively \smash{$H^{\vec{\hat{\theta}}} \{ \ket{\co},\ket{\cl} \}^{\otimes n}$}). Protocol \QOT is shown below.
\begin{protocol}
\begin{enumerate}
\protoline{Preparation}{Alice chooses $\vec{x},\vec{\theta} \in \{0,1\}^n$ at random and sends the $n$ qubits $H^{\vec{\theta}}\ket{\vec{x}}$ to Bob. Bob selects $\vec{\hat{\theta}} \in \{0,1\}^n$ at random and measures the received qubits in basis $\vec{\hat{\theta}}$, obtaining $\vec{\hat{x}}\in\{0,1\}^n$.}
\protoline{Commitment}{\label{item:testset}Bob commits bit-wise to $\vec{\hat{\theta}}$ and $\vec{\hat{x}}$. Alice samples a random subset $t \subset \setn$ of cardinality $k$ and asks Bob to open the commitments to $\hat{\theta}_i$ and $\hat{x}_i$ for all $i\in t$. Alice verifies the opened commitments by checking that $\hat{x}_i = x_i$ whenever $\hat{\theta}_i = \theta_i$. She internally stores the outcome of this check, i.e. \texttt{accept} or \texttt{reject}, for later use in step 4.}
\protoline{Set partitioning}{Alice sends $\vec{\theta}$ to Bob. Bob partitions $\bar{t}$ into the subsets $I_c = \{i \in \bar{t}: \theta_i = \hat{\theta}_i\} $ and $I_{1-c} = \{i \in \bar{t}: \theta_i \ne \hat{\theta}_i\}$ and sends $I_0$ and $I_1$ to Alice.}
\protoline{Key extraction}{
Alice chooses a random $r \in \mathcal{R}$ and sends it to Bob. Bob computes $\vec{\hat{k}}_c = g(r,\vec{\hat{x}}\idx{I_c})$. In case of \texttt{accept}, Alice computes $\vec{k}_0$ and $\vec{k}_1$ as $\vec{k}_0 := g(r,\vec{x}\idx{I_{0}})$ and $\vec{k}_{1} := g(r,\vec{x}\idx{I_{1}})$. 
Otherwise, i.e. in case of \texttt{reject}, she sets $\vec{k}_0$ and $\vec{k}_1$ to random $\ell$-bit strings.}
\end{enumerate}
\caption{\QOT$\!(\perp;c)$}
\end{protocol} 

\noindent
It is trivial to see that for honest Alice and Bob: $\vec{\hat{k}}_c = \vec{k}_c$. Furthermore, security against dishonest Alice, who is trying to learn information on $c$, is easy to see and not the issue here: in case of a perfect BC functionality, Alice learns no information on $c$ no matter what she does; in case of a computationally hiding BC implementation, all information she obtains on $c$ is ``hidden within the commitments'', and thus computational security follows from the computational hiding property. 

Proving security against dishonest Bob is much more subtle, and is the goal of this section. Clearly, {\em if} Bob indeed measures the qubits in the preparation phase with respect to some choice $\vec{\hat{\theta}}$, then security is easy to see: no matter how he partitions $\bar{t}$ into $I_0$ and $I_1$, on at least one of $\vec{x}\idx{I_0}$ and $\vec{x}\idx{I_1}$ he has some lower bounded uncertainty, and privacy amplification finishes the job. The intuition is now that the commitment phase forces Bob to essentially measure all qubits with respect to some choice $\vec{\hat{\theta}}$, as otherwise he will get caught. However, proving this rigorously is non-trivial.

\subsection{The Security Proof}

For our proof of security against dishonest Bob, we first introduce a slightly modified version of the protocol, \QOTx, given below. 
\QOTx is only of proof-technical interest because it asks Alice to perform some actions that she could not do in practice. However, her actions are well-defined, and it follows from standard arguments that Bob's view of \QOT is exactly the same as of \QOTx. It thus suffices to prove security (against dishonest Bob) for \QOTx.


\QOTx is obtained from \QOT by means of the following two modifications. First, for every $i \in \setn$, instead of sending $H^{\theta_i}\ket{x_i}$, Alice prepares an EPR pair $A_i B_i$ of which she sends $B_i$ to Bob and measures $A_i$, at some later point in the protocol, in basis $\theta_i$ to obtain $x_i$. By elementary properties of EPR pairs, and since actions on different subsystems commute, this does not affect Bob's view of the protocol. 
Second, Alice measures her qubits $A\idx{t}$ within the test subset $t$ in {\em Bob's basis} $\vec{\hat{\theta}}_t$ (rather than in $\vec{\theta}_t$) to obtain~$\vec{x}_t$, but she still only verifies correctness of Bob's $\hat{x}_i$'s with $i \in t$ for which $\hat{\theta}_i = \theta_i$. 
Note that by assumption on the BC, the string $\vec{\hat{\theta}}$ to which Bob can open his commitments is uniquely determined at this point, and thus Alice's action is well-defined, although not doable in real life. 
This modification only influences Alice's bits $x_i$ for which $i \in t$ and $\hat{\theta}_i \neq \theta_i$; however, since these bits are not used in the protocol, it has no effect on Bob's view. 

\begin{protocol}
\caption{\QOTx$\!\!(\perp;c)$}
\begin{enumerate}
\protoline{Preparation}{Alice prepares $n$ EPR pairs of the form $(\ket{\co}\ket{\co}+\ket{\cl}\ket{\cl})/\sqrt{2}$, and sends one qubit of each pair to Bob, who proceeds as in the original scheme QOT to obtain $\vec{\hat{\theta}}$ and $\vec{\hat{x}}$. Alice chooses a random $\vec{\theta} \in \set{0,1}^n$, but she does not measure her qubits yet.}
\protoline{Commitment}{Bob commits to $\vec{\hat{\theta}}$ and $\vec{\hat{x}}$, and Alice chooses a random subset $t \subset \setn$ of cardinality $k$, as in \QOT. Next, Alice measures her qubits that are indexed by $t$ in {\em Bob's} basis $\vec{\hat{\theta}}_t$ to obtain~$\vec{x}_t$. Then, Alice sends $t$ to Bob and they proceed as in \QOT, meaning that Bob opens these commitments and Alice verifies them.}
\protoline{Set partitioning}{As in \QOT. Additionally, Alice measures her qubits corresponding to $I_0$ in basis $\vec{\theta}_{I_0}$ to obtain $\vec{x}\idx{I_0}$ and her qubits corresponding to $I_1$ in basis $\vec{\theta}_{I_1}$ to obtain $\vec{x}\idx{I_1}$.}
\protoline{Key extraction}{Exactly as in the original scheme \QOT.}
\end{enumerate}
\end{protocol}

\noindent
Our proof for the security of \QOTx, and thus of \QOT, against dishonest Bob follows quite easily from our treatment of sampling strategies from Section~\ref{sec:quantsampling}. The proof is given below, after the formal security statement in Theorem~\ref{thm:QOT}. 
We would like to point out that our security guarantee against dishonest Bob implies the security definition proposed and studied in~\cite{FS08arxiv} for (randomized) OT, which in particular implies sequential {\em composability} when used as a sub-routine in a classical outer protocol. Note that we do not claim any kind of composability here against dishonest Alice, although when using an ideal BC, sequential composability against dishonest Alice can be proven along similar lines as for QOT in the bounded quantum storage model (see e.g. the arXiv version of \cite{dfrss07}).

\begin{thm}[Security of \QOT]\label{thm:QOT}
Consider an execution of \QOT (respectively \QOTx) between honest Alice and dishonest Bob. Let $\vec{K}_0$ and $\vec{K}_1$ be the keys in $\set{0,1}^\ell$ output by Alice. 
Then, there exists a bit $c$ so that $\vec{K}_{1-c}$ is close to random-and-independent of Bob's view (given $\vec{K}_{c}$) in that for any $\epsilon,\delta > 0$: 
\begin{align*}
\dist\bigl( \rho_{\vec{K}_{1-c}\vec{K}_{c} E}&,{\textstyle\frac{1}{2^\ell}} \id \otimes \rho_{\vec{K}_{c} E} \bigr) \\
&\leq \frac12 \cdot 2^{-\frac12\big(\big( \frac14 - \frac{\epsilon}{2} - \hbin(\delta)\big)(n-k) - \ell \big)} + \sqrt 6 \exp \bigl(- \delta^2 k /100 \bigr) + 2 \exp\bigl(-2\epsilon^2(n-k)\bigr), 
\end{align*}
where $E$ denotes the quantum state output by Bob, and $\id$ the identity operator on $\C^{2^\ell}$. 
\end{thm}

\noindent
On a high level, the proof is as follows. Alice's checking procedure can be understood as applying a sampling strategy to the qubits she holds. From this we obtain that (except with a small error) the joint state she shares with Bob is a superposition of states with small relative Hamming weight within her subsystem~$A\idx{\bar{t}}$. This implies that the joint state is a superposition of states with small relative Hamming weight also within \smash{$A\idx{I_{1-c}}$}, where $c \in \set{0,1}$ is chosen such that $\theta_i \neq \hat{\theta}_i$ for approximately half (or more) of the indices $i$ in $I_{1-c}$. It then follows from Corollary~\ref{cor:minentsup} that $\vec{x}\idx{I_{1-c}}$, obtained by measuring \smash{$A\idx{I_{1-c}}$} in basis $\vec{\theta}\idx{I_{1-c}}$, has high min-entropy, so that privacy amplification concludes the proof.  
The formal proof, which takes care of the details and keeps track of the error term, is given below.

\begin{proof}
We consider the state
\[
\ket{\varphi_{AE_\circ}} \in \H_{A_1} \kron \cdots \kron \H_{A_n} \kron \H_{E_\circ} ,
\]
shared between Alice and Bob, after Bob has committed to $\vec{\hat{\theta}}$ and $\vec{\hat{x}}$, but before Alice chooses the test subset $t$. $\ket{\varphi_{AE_\circ}}$ is obtained from the $n$ EPR-pairs by an arbitrary quantum operation (possibly involving measurements), applied only to Bob's part. Without loss of generality, we may assume that, given the commitments, the joint state is indeed pure. Furthermore, we consider the strings $\vec{\hat{\theta}}$ and $\vec{\hat{x}}$, to which Bob has committed. By the perfectly binding property, these are uniquely determined. For concreteness, and in order to have the notation fit nicely with Section~\ref{sec:quantsampling}, we assume $\vec{\hat{\theta}} = \vec{\hat{x}} = (0,\ldots,0) \in \set{0,1}^n$; however, by Remark~\ref{rem:allzero}, the very same reasoning works for any $\vec{\hat{\theta}}$ and $\vec{\hat{x}}$. 

The crucial observation now is that Alice's checking procedure within the commitment phase of \QOTx can be understood as applying a sampling strategy
to the state $\ket{\varphi_{AE_\circ}}$ in order to test closeness of $A$ to the all-zero state $\ket{\co} \cdots\ket{\co}$. Indeed, Alice chooses a random subset $t \subset \setn$ of cardinality $k$, measures $A\idx{t}$ (in the computational basis) to obtain $\vec{x}\idx{t}$, and decides whether to accept or reject based on $\vec{x}\idx{t}$; specifically, she takes a random subset $s \subseteq t$, given by $s = \Set{i \in t}{\theta_i = \hat{\theta}_i}$, and accepts if and only $x\idx{s} = 0$ for all $i \in s$. This is precisely the sampling strategy $\Psi$ studied in \refex{QOT}, adapted to test closeness to $\ket{\co} \cdots\ket{\co}$ by accepting if and only if $f(t,\vec{x}\idx{t},s)= 0$. Note that, by the random choices of the $\theta_i$'s, $s$ is indeed a random subset of $t$. 

Thus, we can conclude that at the end of the commitment phase, for any fixed $\delta > 0$, the joint state of $A\idx{\bar{t}} E_\circ$ has collapsed to a state $\ket{\psi_{A\idx{\bar{t}} E_\circ}}$ that is (on average over Alice's choice of $t$ and~$s$) $\qep$-close to being a superposition of states with relative Hamming weight at most $\delta$ within $A\idx{\bar{t}}$ (except when Alice rejects the test, but in that case she will output random and independent keys at the end of the protocol and the theorem trivially holds). We proceed by assuming that the state $\ket{\psi_{A\idx{\bar{t}} E_\circ}}$ {\em equals} a superposition of states with small relative Hamming weight, and we book-keep the error $\qep$.%
\footnote{It now follows immediately from Corollary~\ref{cor:minentsup} that $\hmin{X_0 X_1|E_\circ}$ is ``large'', where $X_0$ collects the bits obtained by measuring $A\idx{I_0}$ in basis $\theta\idx{I_0}$, and correspondingly for $X_1$. However, in the end we need that $\hmin{X_{1-c} |X_c E_\circ}$ is ``large'' for some $c$, which does {\em not} follow from the former. Because of that, we need to make a small detour. }
Recall that by Theorem~\ref{thm:quanterr} and \refex{QOT} (and its analysis in Appendix~\ref{app:QOT}), 
$$
\qep \leq \sqrt{\cep} \leq \sqrt 6 \exp\bigl(- k\delta^2 / 100 \bigr)\, .
$$

By the random choices of the $\theta_i$'s, it follows from Hoeffding's inequality (Theorem~\ref{thm:hoeffding}) that 
the Hamming weight of $\vec{\theta}\idx{\bar{t}}$ is lower bounded by $\weight{\vec{\theta}\idx{\bar{t}}} \geq (\frac12 -\epsilon)(n-k)$ except with probability at most \mbox{$2 \exp(-2\epsilon^2(n-k))$}.%
\footnote{Actually, for the one-sided bound, we could save the factor two in front of the $\exp$. }
In the sequel, we assume that the bound holds, and we book-keep the error. 
It follows that regardless of how Bob divides $\bar{t}$ into $I_0$ and $I_1$, there exists $c \in \set{0,1}$ such that $\weight{\vec{\theta}\idx{I_{1-c}}} \geq \frac12(\frac12 -\epsilon)(n-k)$ (if Bob is honest, then $c$ coincides with his input bit). 

By re-arranging Alice's qubits, we write the state $\ket{\psi_{A\idx{\bar{t}} E_\circ}}$ as $\ket{\psi_{A^{1-c} A^c E_\circ}}$, where $A^0 := A\idx{I_0}$ and $A^1 := A\idx{I_1}$. Since $\ket{\psi_{A\idx{\bar{t}} E_\circ}}$ is a superposition of states with Hamming weight at most $(n-k)\delta$ within $A\idx{t}$, it is easy to see that $\ket{\psi_{A^{1-c} A^c E_\circ}}$ is a superposition of states with Hamming weight at most $(n-k)\delta$ within $A^{1-c}$. 
Let the random variables $\vec{X}_{\!1-c}$ and $\vec{X}_{\!c}$ describe the outcome of measuring $A^{1-c}$ and $A^c$ in bases $\vec{\theta}\idx{I_{1-c}}$ and $\vec{\theta}\idx{I_c}$, respectively, and let $\rho_{\vec{X}_{1-c} \vec{X}_c E_\circ}$ be the corresponding hybrid state. We may think of $\rho_{\vec{X}_{1-c} \vec{X}_c E_\circ}$ being obtained by {\em first} measuring $A^{1-c}$, resulting in a hybrid state $\rho_{\vec{X}_{1-c} A^c E_\circ}$, and {\em then} measuring $A^c$; indeed, the order in which these measurements take place have no effect on the final state. 

We can now apply Corollary~\ref{cor:minentsup} to the hybrid state $\rho_{\vec{X}_{1-c} A^c E_\circ}$ obtained from measuring subsystem $A^{1-c}$ within $\ket{\psi_{A^{1-c} A^c E_\circ}}$ and conclude that 
$$
\hmin{\vec{X}_{1-c}|A^c E_\circ} \geq \weight{\vec{\theta}\idx{I_{1-c}}} - \hbin(\delta)\cdot |I_{1-c}| \geq \Big(\frac14 - \frac{\epsilon}{2} - \hbin(\delta)\Big)(n-k) \, .
$$
By a basic property of the min-entropy (``measuring only destroys information"), it follows that the same bound in particular holds for $\hmin{\vec{X}_{1-c}|\vec{X}_c E_\circ}$. 
Applying privacy amplification (Theorem~\ref{thm:PA}), incorporating the error-probabilities (expressed in terms of trace distance) obtained along the proof, and noting that Bob's processing of his information to obtain his final quantum state $E$ does not increase the trace-distance, concludes the proof. 
\end{proof}

\section{Application II: Quantum Key Distribution (QKD)}
\label{sec:QKD}

In quantum key distribution (QKD), Alice and Bob want to agree on a secret key in the presence of an adversary Eve.
Alice and Bob are assumed to be able to communicate over a quantum channel and over an authenticated classical channel.\footnote{If the classical channel between Alice and Bob is not authentic, then authenticity of the communication can still be achieved by information-theoretic authentication techniques, at the cost of requiring Alice and Bob to initially share a short secret key.} Eve may eavesdrop the classical channel (but not insert or modify messages), and she has full control over the quantum channel. 
The first and still most prominent QKD scheme is the famous BB84 QKD scheme due to Bennett and Brassard~\cite{BB84}. 

In this section, we show how our sampling-strategy framework leads to a simple security proof for the BB84 QKD scheme. Proving QKD schemes rigorously secure is a highly non-trivial task, and as such our new proof nicely demonstrates the power of the sampling-strategy framework. Furthermore, \omt{beyond being ``yet another'' security proof for BB84 QKD,} our new proof has some nice features.
For instance, it allows us to explicitly state (a bound on) the error probability of the QKD scheme for any given choices of the parameters. Additionally, our proof does not seem to take unnecessary detours or to make use of ``loose bounds'', and therefore we feel that the bound on the error probability we obtain is rather tight (although we have no formal argument to support this).
 
Our proof strategy can also be applied to other QKD schemes that are based on the BB84 encoding. For example, Lo \etal's QKD scheme\footnote{In this scheme, Alice and Bob bias the choice of the bases so that they measure a bigger fraction of the qubits in the same basis.} \cite{lo05} can be proven secure by following exactly our proof, except that one needs to analyze a slightly different sampling strategy, namely the one from \refex{bias}. On the other hand, it is yet unknown whether our framework can be used to prove e.g. the six-state QKD protocol \cite{bruss1998} secure.

Actually, the QKD scheme we analyze is the entanglement-based version of the BB84 scheme (as initially suggested by Ekert~\cite{Ekert91}). However, it is very well known and not too hard to show that security of the entanglement-based version implies security of the original BB84 QKD scheme. 

The entanglement-based QKD scheme, \QKD, is parametrized by the total number $n$ of qubits sent in the protocol and the number $k$ of qubits used to estimate the error rate of the quantum channel (where we require $k \leq n/2$). Additional parameters, which are determined during the course of the protocol, are the observed error rate $\beta$ and the number $\ell \in \N \cup \set{0}$ of extracted key bits. \QKD makes use of a universal hash function $g:{\cal R} \times \set{0,1}^{n-k} \rightarrow \set{0,1}^\ell$ and a linear binary error correcting code of length $n-k$ that allows to correct up to a $\beta'$-fraction of errors (except maybe with negligible probability) for some $\beta' > \beta$. The choice of how much $\beta'$ exceeds $\beta$ is a trade-off between keeping the probability that Alice and Bob end up with different keys small and increasing the size of the extractable key. We will write $m$ for the bit size of the syndrome of this error-correcting code. Protocol \QKD can be found below.
\newcommand{\preparation}{qubit distribution\xspace}
\newcommand{\Preparation}{Qubit distribution\xspace}
\newcommand{\keyextraction}{key distillation\xspace}
\newcommand{\Keyextraction}{Key distillation\xspace}

\begin{protocol}
\begin{enumerate}
\protoline{\Preparation}{Alice prepares $n$ EPR pairs of the form $(\ket{\co}\ket{\co}+\ket{\cl}\ket{\cl})/\sqrt{2}$, and sends one qubit of each pair to Bob, who confirms the receipt of the qubits. Then, Alice picks random $\vec{\theta} \in \set{0,1}^n$ and sends it to Bob, and Alice and Bob measure their respective qubits in basis $\vec{\theta}$ to obtain $\vec{x}$ on Alice's side respectively $\vec{y}$ on Bob's side.} 
\protoline{Error estimation}{Alice chooses a random subset $s \subset \setn$ of size $k$ and sends it to Bob. Then, Alice and Bob exchange $\vec{x}\idx{s}$ and $\vec{y}\idx{s}$ and compute $\beta := \relwt{\vec{x}\idx{s}\xor\vec{y}\idx{s}}$.}
\protoline{Error correction}{Alice sends the syndrome $syn$ of $\vec{x}\idx{\bar{s}}$ to Bob with respect to a suitable linear error correcting code (as described above). Bob uses $syn$ to correct the errors in $\vec{y}\idx{\bar{s}}$ and obtains $\vec{\hat{x}}\idx{\bar{s}}$. Let $m$ be the bit-size of $syn$.} 
\protoline{\Keyextraction}{Alice chooses a random seed $r$ for a universal hash function $g$ with range $\set{0,1}^\ell$, where $\ell$ satisfies $\ell < (1\!-\!\hbin(\beta))n-k-m$ (or $\ell = 0$ if the right-hand side is not positive), and sends it to Bob. Then, Alice and Bob compute $\vec{k} := g(r,\vec{x}\idx{\bar{s}})$ and $\vec{\hat{k}} := g(r,\vec{\hat{x}}\idx{\bar{s}})$, respectively.}  
\end{enumerate}
\caption{\QKD}
\label{pro:qkd}
\end{protocol}
It is not hard to see that $\vec{k} = \vec{\hat{k}}$ except with negligible probability (in $n$).
Furthermore, if no Eve interacts with the quantum communication in the \preparation phase then $\vec{x} = \vec{y}$ in case of a noise-free quantum channel, or more generally, $\relwt{\vec{x}-\vec{y}} \approx \phi$ in case the quantum channel is noisy and introduces an error probability $0 \leq \phi < \frac12$. It follows that $\beta \approx \phi$, so that using an error correcting code that approaches the Shannon bound, Alice and Bob can extract close to $(1-2\hbin(\phi))(n-k)$ bits of secret key, which is positive for $\phi$ smaller than approximately $11\%$. 
The difficult part is to prove security against an active adversary Eve. We first state the formal security claim. 

Note that we cannot expect that Eve has (nearly) no information on $\vec{K}$, i.e. that $\dist\bigl(\rho_{\vec{K}E},\frac{1}{|{\cal K}|}\id_{\vec{K}}\otimes\rho_E\bigr)$ is small, since the bit-length $\ell$ of $\vec{K}$ is not fixed but depends on the course of the protocol, and Eve can influence and thus obtain information on $\ell$ (and thus on $\vec{K}$). Theorem~\ref{thm:QKD} though guarantees that the bit-length $\ell$ is the {\em only} information Eve learns on $\vec{K}$, in other words, $\vec{K}$ is essentially random-and-independent of $E$ when given $\ell$.
\begin{thm}[Security of \QKD]\label{thm:QKD}
Consider an execution of \QKD in the presence of an adversary Eve. Let $\vec{K}$ be the key obtained by Alice, and let $E$ be Eve's quantum system at the end of the protocol. Let $\vec{\tilde{K}}$ be chosen uniformly at random of the same bit-length as $\vec{K}$. Then, for any $\delta$ with $\beta+\delta \leq \frac12$:
\[
\dist\bigl( \rho_{\vec{K} E^{ \atop }},\rho_{\vec{\tilde{K}} E} \bigr) 
  \leq \frac12 \cdot 2^{-\frac12\big(\big( 1 - \hbin(\beta+\delta)\big)n-k - m - \ell \big)} +
2 \exp \bigl(-\textstyle \frac16 \delta^2 k\bigr) \, .
\]
\end{thm}

\noindent
From an application point of view, the following question is of interest. Given the parameters $n$ and $k$, and given a course of the protocol with observed error rate $\beta$ and where an error-correcting code with syndrome length $m$ was used, what is the maximal size $\ell$ of the extractable key $\vec{K}$ if we want $\dist( \rho_{\vec{K} E^{ \atop }},\rho_{\vec{\tilde{K}} E}) \leq \epsilon$ for a given $\epsilon$? From the bound in Theorem~\ref{thm:QKD}, it follows that for every choice of $\delta$ (with $\beta+\delta \leq \frac12$), one can easily compute a possible value for $\ell$ simply by solving for $\ell$.  In order to compute the optimal value, one needs to maximize $\ell$ over the choice of $\delta$. 

The formal proof of Theorem~\ref{thm:QKD} is given below. Informally, the argument goes as follows. The error estimation phase can be understood as applying a sampling strategy. From this, we can conclude that the state from which the raw key, $\vec{x}\idx{\bar{s}}$, is obtained, is a superposition of states with bounded Hamming weight, so that Corollary~\ref{cor:minentsup} guarantees a certain amount of min-entropy within $\vec{x}\idx{\bar{s}}$. Privacy amplification then finishes the proof. 

To indeed be able to model the error estimation procedure as a sampling strategy, we will need to consider a modified but {\em equivalent} way for Alice and Bob to jointly obtain $\vec{x}\idx{s}$ and $\vec{y}\idx{s}$ from the initial joint state, which will allow them to obtain the {\sc xor}-sum $\vec{x}\idx{s}\xor\vec{y}\idx{s}$, and thus to compute $\beta$, {\em before} they measure the remaining part of the state, whose outcome then determines $\vec{x}\idx{\bar{s}}$. This modification is based on the so-called {\sc cnot} operation, $\CNOT$, acting on $\C^2 \otimes \C^2$, and its properties that
\begin{equation}\label{eq:cnot}
\CNOT(\ket{b}\ket{c}) = \ket{b}\ket{b\oplus c}
\qquad\text{and}\qquad
\CNOT(H\ket{b}H\ket{c}) = H\ket{b\oplus c}H\ket{c} \, ,
\end{equation}
where the first holds by definition of $\CNOT$, and the second is straightforward to verify.

\def\SYN{S\hspace{-0.15ex}Y\hspace{-0.5ex}N}
\begin{proof}
Throughout the proof, we use capital letters, $\vec{\Theta}$, $\vec{X}$ etc. for the {\em random variables} representing the corresponding choices of $\vec{\theta}$, $\vec{x}$ etc.\ in protocol \QKD.
Let the state, shared by Alice, Bob and Eve right after the quantum communication in the \preparation phase, be denoted by $\ket{\psi_{ABE_\circ}}$;\footnote{Note that $E_\circ$ represents Eve's quantum state just after the quantum communication stage, whereas $E$ represents Eve's entire state of knowledge at the end of the protocol (i.e., the quantum information and all classical information gathered during execution of \QKD).} without loss of generality, we may indeed assume the shared state to be pure. 
For every $i\in \setn$, Alice and Bob then measure the respective qubits $A_i$ and $B_i$ from $\ket{\psi_{ABE_\circ}}$ in basis $\Theta_i$, obtaining $X_i$ and $Y_i$.
This results in the hybrid state $\rho_{\vec{\Theta XY}E_\circ}$. 
For the proof, it will be convenient to introduce the additional random variables $\vec{W} = (W_1,\ldots,W_n)$ and $\vec{Z} = (Z_1,\ldots,Z_n)$, defined by
\begin{equation}\label{eq:WandZ}
Z_i := X_i \xor Y_i
\qquad\text{and}\qquad
W_i := \left\{\begin{array}{cl}
X_i & \text{if $\Theta_i = 0$} \\
Y_i & \text{if $\Theta_i = 1$} 
\end{array}\right. \, .
\end{equation}
Note that, when given $\vec{\Theta}$, the random variables $\vec{W}$ and $\vec{Z}$ are uniquely determined by $\vec{X}$ and $\vec{Y}$ {\em and vice versa}, and thus we may equivalently analyze the hybrid state $\rho_{\vec{\Theta WZ}E_\circ}$. 

For the analysis, we will consider a slightly {\em different} experiment for Alice and Bob 
to obtain the very {\em same} state $\rho_{\vec{\Theta WZ}E_\circ}$; the advantage of the modified experiment is that it can be understood as a sampling strategy. 
%
%
%
%
The modified experiment is as follows.
First, the \textsc{cnot} transformation is applied to every qubit pair $A_iB_i$ within $\ket{\psi_{ABE_\circ}}$ for $i\in \setn$, such that the state 
$\ket{\varphi_{ABE_\circ}} = (\CNOT^{\kron n} \kron \mathbb{I}_{E_\circ}) \ket{\psi_{ABE_\circ}}$ is obtained. 
Next, $\vec{\Theta}$ is chosen at random as in the original scheme, and for every $i \in \setn$ the qubit pair $A_i B_i$ of the transformed state is measured as in the original scheme depending on $\Theta_i$; however, if $\Theta_i = 0$ then the resulting bits are denoted by $W_i$ and $Z_i$, respectively, and if $\Theta_i = 1$ then they are denoted by $Z_i$ and $W_i$, respectively, such that which bit is assigned to which variable depends on $\Theta_i$. This is illustrated in Figure~\ref{fig:QKD} (left and middle), where light and dark colored ovals represent measurements in the computational and Hadamard basis, respectively.
It now follows immediately from the properties \eqref{eq:cnot} of the CNOT transformation and from the relation \eqref{eq:WandZ} between $\vec{X},\vec{Y}$ and $\vec{W},\vec{Z}$ that the state $\rho_{\vec{\Theta WZ}E_\circ}$ (or, equivalently, $\rho_{\vec{\Theta XY}E_\circ}$) obtained in this modified experiment is exactly the same as in the original. 

\begin{figure}[H]
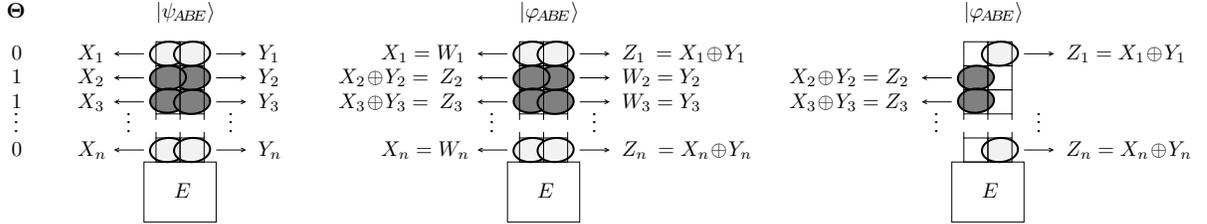

 \centering 
 \scalebox{0.8}{ \input{ps/EntanglementChecking1.pstex_t} }
 \qquad
 \scalebox{0.8}{ \input{ps/EntanglementChecking2.pstex_t} }
 \qquad\qquad\quad
 \scalebox{0.8}{ \input{ps/EntanglementChecking3.pstex_t} }
 \qquad\qquad
 \caption{Original and modified experiments for obtaining the same state $\rho_{\vec{\Theta WZ}E_\circ}$. }
 \label{fig:QKD} 
\end{figure}

An additional modification we may do without influencing the final state is to {\em delay} some of the measurements: we assume that first the qubits are measured that lead to the $Z_i$'s, and only at some later point, namely after the {\em error estimation} phase, the qubits leading to the $W_i$'s are measured (as illustrated in Figure~\ref{fig:QKD}, right). This can be done since the relative Hamming weight of $X_S \oplus Y_S$ for a random subset $S \subset \setn$ (of size $k$) can be computed given $\vec{Z}$ alone. 

The crucial observation is now that this modified experiment can be viewed as a particular sampling strategy $\Psi$, as a matter of fact as the sampling strategy discussed in Example~\ref{ex:qkdsestrat}, being applied to systems $A$ and $B$ of the state $\ket{\varphi_{ABE_\circ}}$. Indeed: first, a subset of the $2n$ qubit positions is selected according to some probability distribution, namely of each pair $A_i B_i$ one qubit is selected at random (determined by $\Theta_i$). Then, the selected qubits are measured to obtain the bit string $\vec{Z} = (Z_1,\ldots,Z_n)$. And, finally, a value $\beta$ is computed as a (randomized) function of $\vec{Z}$: $\beta = \relwt{\vec{Z}\idx{S}}$ for a random $S \subset \setn$ of size $k$. 
We point out that here the reference basis (as explained in~\refrem{allzero}) is not the computational basis for all qubits, but is the Hadamard basis on the qubits in system $A$ and the computational basis in system $B$; however, as discussed in Remark~\ref{rem:allzero}, we may still apply the results from Section~\ref{sec:quantsampling} (appropriately adapted). 

It thus follows that for any fixed $\delta>0$, the remaining state, from which $\vec{W}$ is then obtained, is (on average over $\vec{\Theta}$ and $S$) \qep-close to a state which is (for any possible values for $\vec{\Theta}$, $\vec{Z}$ and~$S$) a superposition of states with relative Hamming weight in a $\delta$-neighborhood of $\beta$. Note that the latter has to be understood with respect to the fixed reference basis (i.e., the Hadamard basis on $A$ and the computational basis on $B$). In the following, we assume that the remaining state {\em equals} such a superposition, but we remember the error
$$
\qep \leq \sqrt{\cep} \leq 2 \exp \bigl(-\textstyle \frac16 \delta^2 k\bigr) \, .
$$
where the bound on \cep is derived in \ifthenelse{\boolean{includeFifthEx}}{Appendix~\ref{app:qkdse}}{Appendix~\ref{app:qkdsestrat}}.

Recall that $\vec{W}$ is now obtained by measuring the remaining qubits; however, the basis used is opposite to the reference basis, namely the computational basis on the qubits $A_i$ and the Hadamard basis on the qubits $B_i$. 
Hence, by \refcor{minentsup} (and the subsequent discussion) we get a lower bound on the \minent of \vec{W}: 
\[
\hmin{\vec{W}|\vec{\Theta Z} S E_\circ}\geq (1 - \hbin(\beta + \delta)) n \, .
\]
Since $\vec{W}$ is uniquely determined by $\vec{X}$ (and vice versa) when given $\vec{\Theta}$ and $\vec{Z}$, the same lower bound also holds for $\hmin{\vec{X}|\vec{\Theta Z} S E_\circ}$. 
Note that in \QKD, the $k$ qubit-pairs that are used for estimating $\beta$ are not used anymore in the \keyextraction phase, so we are actually interested in the \minent of $\vec{X}\idx{\!\bar{S}}$. Additionally, we should take into account that Alice sends an $m$-bit syndrome $\SYN$ during the error correction phase. Hence, by using the chain rule, we obtain
\[
\hmin{\vec{X}\idx{\!\bar{S}}|\vec{\Theta Z} \vec{X}\idx{\!S} \SYN E_\circ} \geq (1 - \hbin(\beta + \delta)) n-k-m.\footnote{Probably, it is possible to prove the lower bound: $(1-\hbin(\beta+\delta))(n-k)-m$ using a different sampling strategy. However, for that case the error probability of the related classical sampling strategy becomes harder to analyze. We have chosen for the current proof strategy and bound for the sake of simplicity.} 
\]
Finally, we apply privacy amplification (\refthm{PA}) which concludes the proof.
\end{proof}


\section{Conclusion}

We have shown a framework for predicting some property (namely the approximate Hamming weight, appropriately defined) of a population of quantum states, by measuring a small sample subset. The framework allows for new and simple security proofs for important quantum cryptographic protocols: the Bennett \etal QOT and the BB84 QKD scheme. 
We find it particularly interesting that with our framework, the protocols for QOT and QKD can be proven secure by means of very similar techniques, even though they implement fundamentally different cryptographic primitives, and are intuitively secure due to very different reasons (namely in QOT the commitments force Bob to measure the communicated qubits, whereas in QKD Eve disturbs the communicated qubits when trying to observe them).

\section{Acknowledgments}
We would like to thank Dejan Dukaric for spotting a minor flaw in one of the sampling-strategy error-probability calculations, Dominique Unruh for his helpful comments regarding composability, as well as Severin Winkler for pointing out an issue related to abort in our QOT protocol. 

\bibliography{quantum}


\begin{appendix}

\section{A Two-Sided Version of Serfling's Bound}
\label{app:serfling}
Let $\vec{b} \in \set{0,1}^n$ be a bit string with relative Hamming weight $\mu = \relwt{\vec{b}}$. Let the random variables $Y_1, Y_2,\ldots, Y_k$ be obtained by sampling $k$ random entries from $\vec{b}$ {\em without replacement}. Then, a result by Serfling \cite{serfling74} says that for any $\delta > 0$, the random variable $\bar{Y}:=\frac{1}{k}\sum_i Y_i$ satisfies
\[
\Pr[ \bar Y -\mu > \delta] \leq \exp\Big( \frac{-2 \delta^2 k n}{n-k+1} \Big).
\]
The aim of this section is to prove that the one-sided bound above implies the two-sided bound 
\[
\Pr[|\bar Y -\mu| > \delta] \leq 2\exp\Big(  \frac{-2 \delta^2 k n}{n-k+1} \Big).
\]
\begin{proof}
\begin{align*}
\Pr[| \bar Y -\mu| > \delta] &=  \Pr[| \tfrac1k \sum_i Y_i -\mu| > \delta] \\&=  \Pr[ \tfrac1k \sum_i Y_i -\mu < -\delta] + \Pr[ \tfrac1k \sum_i Y_i -\mu > \delta]\\
&= \Pr[ \tfrac1k \sum_i (Y_i - 1)  - (\mu - 1) < -\delta] + \Pr[ \tfrac1k \sum_i Y_i -\mu > \delta]\\
&= \Pr[ \tfrac1k \sum_i (1- Y_i )  - (1- \mu ) > \delta] + \Pr[ \tfrac1k \sum_i Y_i -\mu > \delta]
\end{align*}
Note that the random variable $1-Y_i$ has mean $1-\mu$.
We can now apply Serfling's bound twice, yielding the claim.
\end{proof}
\section{Error Probabilities of the Example Sampling Strategies}
\label{app:EPs}

\subsection{Example 1 --- Random sampling \emph{without} replacement}
\label{app:ex1}
It follows immediately from \refthm{hoeffding} that the estimate is $\delta$-close to the relative Hamming weight $ \relwt{\vec{q}}$ of $\vec{q}$ except with probability at most $2\exp(-2\delta^2 k)$. However, we want to analyze closeness of the estimate to $\relwt{\vec{q}\idx{\bar{T}}}$ (still treating $T$ as a random variable). 
This can be derived easily as follows. We can write
$\relwt{\vec{q}} = \alpha \relwt{\qT} + (1-\alpha) \relwt{\qbT}$, 
where $\alpha := k/n$, and thus can see that
$$
\relwt{\vec{q}\idx{\bar{T}}} - \relwt{\vec{q}\idx{T}} = \frac{1}{1-\alpha}\Big(\relwt{\vec{q}} - \alpha \relwt{\qT} \Big) - \relwt{\vec{q}\idx{T}} = \frac{1}{1-\alpha}\Big(\relwt{\vec{q}} - \relwt{\qT} \Big)
$$
so that
\begin{align}
\nonumber \cep  &= \max_{\vec{q}} \Pr\Bigl[\vec{q} \notin B_{T,S}^\delta \Bigr] 
= \max_{\vec{q}} \Pr\bigl[\left| \relwt{\vec{q}\idx{\bar{T}}} - \relwt{\vec{q}\idx{T}}\right| \geq \delta \bigr] \\
&= \max_{\vec{q}} \Pr\bigl[\left| \relwt{\vec{q}} - \relwt{\vec{q}\idx{T}}\right| \geq (1\!-\!\alpha)\delta \bigr] 
\leq 2\exp\bigl(-2(1\!-\!\alpha)^2\delta^2 k\bigr).
\label{eq:allk}
\end{align}
Under assumption of $k\leq n/2$, we obtain a simple bound for the latter expression,
\begin{equation}
\cep \leq 2\exp\bigl(-2(1\!-\!\alpha)^2\delta^2 k\bigr)  \leq 2 \exp\bigl(\textstyle-\frac12\delta^2 k\bigr).
\label{eq:simple}
\end{equation}
We obtain the following bound if we use the bound from \cite{serfling74}:
\begin{align*}
\cep  &= \max_{\vec{q}} \Pr\bigl[\left| \relwt{\vec{q}} - \relwt{\vec{q}\idx{T}}\right| \geq (1\!-\!\alpha)\delta \bigr] 
\\&\leq 2\exp\bigl(-\textstyle\frac{2(1\!-\!\alpha)^2\delta^2 k n}{n-k+1}\bigr) 
= 2\exp\bigl( -\frac{2 k (n-k)^2 \delta ^2}{n (n-k+1)} \bigr) \leq 
2\exp\bigl( - \frac{\delta ^2 kn}{n+2} \bigr). 
\end{align*}
for $k\leq n/2$, because $-\frac{2 k (n-k)^2 \delta ^2}{n (n-k+1)} $ is convex in $k$, and $- \frac{\delta ^2 kn}{2+n}$ is linear in $k$ and equality holds at $k=0$ and $k=n/2$, hence it is a tight linear upper bound.
\subsection{Example 2 --- Random sampling \emph{with} replacement}
Computing the error probability for \refex{withreplacement} actually turns out to be tricky. 
Although, as in Example~1 above, \refthm{hoeffding} applies and guarantees that the estimate is likely to be close to $\relwt{\vec{q}}$, showing that the estimate is likely to be close to $\relwt{\qbT}$ seems to be non-trivial here. Since we make no further use of this example sampling strategy, we refrain from analyzing its error probability. 
 
\subsection{Example 3 --- Uniformly random subset sampling}
\label{app:ex3}
Note that for any fixed choice $k = |t|$, $t$ is obtained as in random sampling without replacement. Because $t$ is sampled uniformly at random, the expectation of $k$ is given by $E[k]=n/2$. Hence, by making use of Hoeffding's inequality, we can say that for $0< \beta <\frac12$, $\Pr [|\frac{k}{n}-\frac{1}{2}| \geq \beta]\leq 2 \exp (-2\beta^2 n)$.

Informally, the idea is to start off with an upper bound on \cep obtained in Appendix~\ref{app:ex1} (the case of sampling without replacement), and transform it into an upper bound that holds under the assumption that $k \in [(\frac12-\beta)n, (\frac12 + \beta)n]$.
Note that we cannot use the simple bound \refeq{simple} from Appendix~\ref{app:ex1}, because that result was obtained under the assumption that $k \leq n/2$, 
and here this assumption does not hold. Instead, we use bound \refeq{allk} from Appendix~\ref{app:ex1},
\begin{equation}
\label{eq:bnd}
\cep \leq 2 \exp \Big(-2\big(1-\tfrac{k}{n}\big)^2 \delta^2 k \Big)
\end{equation}
which \emph{does} hold for all $k \in \{0, \ldots, n\}$.

To get an upper bound for \refeq{bnd}, we replace the first occurrence of $k$ in that expression (in the numerator of the fraction) by an upper bound for $k$, and the second occurrence of $k$ by a lower bound for $k$. The upper and lower bound for $k$ are simply given by the (appropriate) boundary points of the interval $[(\frac12-\beta)n, (\frac12 + \beta)n]$. I.e.,
\[
2 \exp \Big(-2 n\delta^2 \Big(1- \frac{ (\frac12 +\beta ) n}{n}\Big)^2  (\tfrac12-\beta)\Big) = 2\exp\big(  -2 n \delta^2(\tfrac{1}{2}-\beta )^3 \big)
\]
To compute \cep, we use a union bound to combine the upper bound above, which holds under assumption that $k$ lies inside the previously defined interval, with the upper bound on the probability that $k$ does \emph{not} lie in this interval, 
\[
\cep \leq 2 \exp\Big(\!-2n\delta^2 \big(\tfrac12 - \beta\big)^3\Big)+ 2 \exp (-2\beta^2 n).
\]
Setting $\beta=\delta/4$ in the expression above yields $-n \delta^2 (2-\delta)^3 /32$ for the exponent of the first summand, and $-n\delta^2/8$ for the exponent of the second summand. Because $0<\delta<1$ (\refdef{class-err-prob}), a suitable upper bound for both exponents is $-n \delta^2 /32$.\footnote{Note that our goal is to find a short and simple expression, rather than finding the tightest bound.} This gives the following simpler bound, 
\begin{align*}
\cep &\leq  4 \exp(- n \delta^2/32).
\end{align*}

\subsection{Example 4 --- \exname{}}
\label{app:QOT}
From Appendix~\ref{app:ex1}, we know that
$\Pr\bigl[|\relwt{\qbT}-\relwt{\qT}| \geq \xi\bigr] \leq 2\exp(-\frac{1}{2}\xi^2k)$, for $k<n/2$. 
Additionally, the selection of the seed $s$ and the computation of $\func$ can be viewed as applying uniformly random subset sampling to $\vec{q}\idx{t}$. Hence, it follows from Appendix~\ref{app:ex3} that 
$\max_{\vec{q}} \Pr\bigl[|\relwt{\qT} - \relwt{\vec{q}\idx{S}}|\geq \gamma \bigr] \leq 4 \exp(- k \gamma^2 /32 )$.
Setting $\delta = \xi+\gamma$, and using triangle inequality and union bound, we 
obtain
\begin{align*}
\cep &= \max_{\vec{q}} \Pr\bigl[|\relwt{\vec{q}\idx{S}} -\relwt{\qbT}|\geq  \delta \bigr] 
\\&\leq \min_{0<\xi<\delta}\Bigl[\textstyle 2\exp\bigl(-\frac{1}{2}\xi^2k\bigr) + 4 \exp\bigl(- k(\delta-\xi)^2 / 32 \bigr)   \Bigr] \\ 
& \leq 6 \exp\bigl(- k\delta^2 / 50 \bigr),
\end{align*}
where the last inequality follows from setting $\xi = \delta / 5$ such that the two exponents coincide.
\ifthenelse{\boolean{includeFifthEx}}{
\subsection{Example 5 --- \qkdexname{}}
\label{app:qkdse}
\qkdseerror
}{}

\subsection{Example 6 --- \biasexamplename}
\label{app:ex6}
It will be convenient to define the index set $t$ as the union of two subsets, $t_0 \subset \setn \times \set{0}$ and $t_1 \subset \setn \times \set{1}$. 
Note that the complements of these subsets should now be understood as $\bar{t}_0 = (\setn \times \set{0}) \setminus t_0$ and $\bar{t}_1 = (\setn \times \set{1}) \setminus t_1$.
Let $t_0$ and $t_1$ be constructed as follows. We first sample a set $\tilde{t} \subset \setn $; for each element of $\setn$, we include it in $\tilde{t}$ with probability $p$. Then, $t_0 := \tilde{t} \times \set{0}$ and  $t_1 := (\setn \setminus \tilde{t}) \times \set{1}$. 
Like $t$, the seed $s$ is also defined as the union of two randomly chosen sets, $s = s_0 \cup s_1$, where $s_0 \subset t_0$ and $s_1 \subset t_1$.\footnote{Again, \refrem{container} applies.} These sets have fixed size; for a parameter $k\in \mathbb{N}$, $|s_0|= \frac{k}{2}$ and $|s_1|=\frac{k}{2}$.
Now, the estimate for \relwt{\vec{q}\idx{\bar{t}}} is computed as $f(t,\vec{q}_{t},s) = \frac1n \bigl( |\bar{t}_0|\ \relwt{\vec{q}_{s_0}} +  |\bar{t}_1|\ \relwt{\vec{q}_{s_1}} \bigr)$.

We need to show that $\relwt{\vec{q}\idx{\bar{T}}}$ is likely to be close to $\relwt{\vec{q}\idx{S}}$. Because we compute an estimate for \relwt{\vec{q}\idx{\bar{T}}} as a function of \relwt{\vec{q}\idx{S_0}} and \relwt{\vec{q}\idx{S_1}}, we will first show that (with high probability) $\relwt{\vec{q}\idx{T_0}} \approx \relwt{\vec{q}\idx{S_0}}$ and $\relwt{\vec{q}\idx{T_1}} \approx \relwt{\vec{q}\idx{S_1}}$. Then, we argue that $\relwt{\vec{q}\idx{\bar{T_0}}} \approx \relwt{\vec{q}\idx{T_0}}$ and $\relwt{\vec{q}\idx{\bar{T_1}}} \approx \relwt{\vec{q}\idx{T_1}}$, from which we can also conclude (using the union bound) that $\relwt{\vec{q}\idx{\bar{T_0}}} \approx \relwt{\vec{q}\idx{S_0}}$ and $\relwt{\vec{q}\idx{\bar{T_1}}} \approx \relwt{\vec{q}\idx{S_1}}$. Finally, we apply the union bound again and combine the two bounds to obtain an upper bound for $\Pr \bigl[| \relwt{\vec{q}_{\bar{T}}} - \frac1n( |\bar{T}_0|\ \relwt{\vec{q}_{S_0}} + |\bar{T}_1 |\  \relwt{\vec{q}_{S_1}}) | \geq \delta \bigr]$.

The first step in the proof 
follows directly from Hoeffding's inequality,
\begin{align*}
\Pr\bigl[ \left| \relwt{\vec{q}_{T_0}} - \relwt{\vec{q}\idx{S_0}} \right| \geq \gamma \bigr] \leq 2\exp\left(- 2 |S_0| \gamma^2 \right) =  2\exp\bigl(-  k \gamma^2 \bigr) ,\quad \text{for any }\gamma >0.
\end{align*}
Trivially, this bound also applies to the relation between \relwt{\vec{q}_{T_1}} and \relwt{\vec{q}\idx{S_1}}, if we substitute appropriately.
The second step, showing that $\relwt{\bar{T}_0}$ (respectively $\relwt{\bar{T}_1}$) is likely to be close to $\relwt{T_0}$ (resp. $\relwt{T_1}$), is slightly more involved. Namely, although the sum of the sizes of $T_0$ and $T_1$ is constant (to be precise, $|T_0|+|T_1|=n$), their individual sizes are random. In \refex{coinflip} (see also Appendix~\ref{app:ex3}), we have already encountered a similar, though not identical, situation, i.e., \refex{coinflip} considers uniformly random one-out-of-two sampling whereas here we analyze one-out-of-two sampling according to a Bernoulli $(p,1\!-\!p)$ distribution. 
Nonetheless, it is straightforward to generalize the proof of Appendix~\ref{app:ex3} to this (more general) case. 

Let $X:=|T_0|$. The expectation of $X$ is given by $E[X]=np$. 
Let $\mathcal{E}$ be the event that $X \in [(p - \beta)n , (p + \beta)n]$, for $ \beta>0 $. From Hoeffding's inequality, we known that $\Pr[\bar{\mathcal{E}} ] = \Pr [|\frac{X}{n} - p | \geq \beta]\leq 2 \exp (-2\beta^2 n)$. Like in Appendix~\ref{app:ex3}, we find an upper bound that holds conditioned on the event $\mathcal{E}$, by substituting the boundary points of the interval used to define $\mathcal{E}$ in \refeq{bnd},
\begin{align*}
\Pr \bigl[| \relwt{\vec{q}_{T_0}} - \relwt{\vec{q}_{\bar{T}_0}} | \geq \delta \ \big|\ \mathcal{E}\bigr] &\leq-2 (p-\beta)n \left(1-\frac{(p+\beta)n}{n}\right)^2 \\&= 2 \exp \bigl( -2 n\delta^2 (1-p-\beta)^2 (p-\beta) \bigr).
\end{align*}
Next, we apply the union bound to show that for $0<\epsilon<\gamma$ 
\begin{align*}
 \Pr\bigl[\left| \relwt{\vec{q}\idx{\bar{T}_0}} - \relwt{\vec{q}\idx{S_0}} \right| \geq \gamma \ \big|\ \mathcal{E} \bigr] & \leq 2 \exp \bigl(-2 n\epsilon^2 (1-p-\beta)^2 (p-\beta) \bigr) + 2\exp\left(-  k  (\gamma - \epsilon)^2 \right) 
\end{align*}
By substituting $p$ by $1-p$ in the expression above, we also obtain
\begin{align*}
 \Pr\bigl[\left| \relwt{\vec{q}\idx{\bar{T}_1}} - \relwt{\vec{q}\idx{S_1}} \right| \geq \gamma \ \big|\ \mathcal{E} \bigr] & \leq 2 \exp \bigl(-2 n\epsilon^2 (p-\beta)^2 (1-p-\beta) \bigr) + 2\exp\left(-  k  (\gamma - \epsilon)^2 \right) 
\end{align*}
Finally, we combine the two bounds and we get rid of the conditioning on $\mathcal{E}$ by adding $\Pr[\bar{\mathcal{E}}]$. For any $\delta > 0$ and $0 < \epsilon < \delta$, we may write
\begin{align*}
\cep &= \max_{\vec{q}}  \Pr \bigl[|  \relwt{\vec{q}_{\bar{T}}} - \frac1n( |\bar{T}_0|\ \relwt{\vec{q}_{S_0}} +|\bar{T}_1| \  \relwt{\vec{q}_{S_1}}) | \geq \delta \bigr]\\
&=\max_{\vec{q}} \Pr \bigl[| \weight{\vec{q}_{\bar{T}}} -  |\bar{T}_0|\ \relwt{\vec{q}_{S_0}} +|\bar{T}_1| \  \relwt{\vec{q}_{S_1}} | \geq n\delta \bigr]\\
&=\max_{\vec{q}} \Pr \bigl[| \weight{\vec{q}_{\bar{T}}} -  |\bar{T}_0|\ \relwt{\vec{q}_{S_0}} +|\bar{T}_1| \  \relwt{\vec{q}_{S_1}} | \geq (|\bar{T}_0| \delta +|\bar{T}_1| \delta )  \bigr]\\
&\leq \max_{\vec{q}}\Pr\bigl[\left| \relwt{\vec{q}\idx{\bar{T}_0}} - \relwt{\vec{q}\idx{S_0}} \right| \geq \delta \bigr] + \Pr\bigl[\left| \relwt{\vec{q}\idx{\bar{T}_1}} - \relwt{\vec{q}\idx{S_1}} \right| \geq \delta\bigr] \\
& \leq 2 \exp \bigl(-2 n\epsilon^2 (1-p-\beta)^2 (p-\beta) \bigr) + 2 \exp \bigl(-2 n\epsilon^2 (p-\beta)^2 (1-p-\beta) \bigr) + \ldots \\
& \phantom{\leq } + 4\exp\left(-  k  (\delta -\epsilon )^2 \right) + 2 \exp (-2\beta^2 n)
\end{align*}

\section{Proof of \reflem{minentsup}}\label{app:minentsup}

\begin{proof}
We will show that $|J| \rho_{WE}^\mathrm{mix} \geq \rho_{WE}$, to be understood in that $|J|\rho_{WE}^\mathrm{mix} - \rho_{WE}$ is positive semi-definite. With this shown, it then follows that for any density matrix $\sigma_E$ and for any non-negative $h \in \R$ 
$$
2^{-(h-\log|J|)} \cdot \id_W \kron \sigma_E - \rho_{WE} \geq
2^{-h} |J| \cdot \id_W \kron \sigma_E - |J|\rho_{WE}^\mathrm{mix} = |J| \big( 2^{-h} \cdot \id_W \kron \sigma_E - \rho_{WE}^\mathrm{mix} \big) 
$$
so that if the right-hand side is positive semi-definite then so is the left-hand side. 
The claimed bound $\hmin{\rho_{WE}|E} \geq \hmin{\rho_{WE}^\mathrm{mix}|E} - \log|J|$ then follows by the definition of the min-entropy. 

\noindent Writing out the measurements explicitly yields
\begin{align*}
\rho_{WE} &= \sum_{w\in \mathcal{W}} (\outs{w} \kron \id_E) \outs{\varphi_{AE}} (\outs{w} \kron \id_E)
= \sum_{w\in \mathcal{W}} \sum_{{i},{j} \in J} \alpha_{i} \bar{\alpha}_{j} \ket{w}\inprod{w}{{i}}\inprod{{j}}{w} \tvec{w} \kron \out{\varphi^{i}_E}{\varphi^{j}_E}
\end{align*}
and
\[
\rho_{WE}^\mathrm{mix} = \sum_{{i}\in J} |\alpha_{i}|^2 \sum_{w\in \mathcal{W}}|\inprod{w}{{i}}|^2\outs{w}\kron \outs{\varphi_E^{i}}.
\]

\noindent
We want to show that $\bra{\xi} (|J|\rho_{WE}^\mathrm{mix}-\rho_{WE}) \ket{\xi} \geq 0$ for all $\ket{\xi} \in \H_W \kron \H_E$. We first consider $\ket{\xi}$ of the special form $\ket{\xi} = \ket{v}\ket{\psi_E}$ with $v \in \mathcal{W}$, 
%
%
%
%
and compute/bound $\tvec{\xi} \rho_{WE} \ket{\xi}$ and $\tvec{\xi} \rho_{WE}^\mathrm{mix} \ket{\xi}$ as 
\begin{align*}
\tvec{\xi} \rho_{WE} \ket{\xi} &=\sum_{{i},{j} \in J} \alpha_{i} \bar{\alpha}_{j} \inprod{v}{{i}}\inprod{{j}}{v}  \inprod{\psi_E}{\varphi^{i}_E} \inprod{\varphi^{j}_E}{\psi_E} = \Big( \sum_{{i}\in J} \alpha_{i}  \inprod{v}{{i}}  \inprod{\psi_E}{\varphi^{i}_E}\Big)\Big( \sum_{{j}\in J} \bar{\alpha}_{j}\inprod{{j}}{v} \inprod{\varphi^{j}_E}{\psi_E}\Big)\\
&=\Big| \sum_{{i}\in J} \alpha_{i}\inprod{v}{{i}} \inprod{\psi_E}{\varphi^{i}_E}\Big|^2,
\end{align*}
and
\begin{align*}
\tvec{\xi} \rho_{WE}^\mathrm{mix} \ket{\xi} &=\sum_{{i} \in J} |\alpha_{i}|^2  | \inprod{v}{{i}}|^2  |\inprod{\psi_E}{\varphi^{i}_E}|^2 \geq \frac{1}{|J|} \Big| \sum_{i\in J} \alpha_{i}  \inprod{v}{{i}}  \inprod{\psi_E}{\varphi^{i}_E}\Big|^2 =\frac{1}{|J|} \tvec{\xi} \rho_{WE} \ket{\xi},   
\end{align*}
where the inequality follows from Cauchy-Schwarz inequality.
The claim, $\bra{\xi} (|J|\rho_{WE}^\mathrm{mix}-\rho_{WE}) \ket{\xi} \geq 0$, for an {\em arbitrary} $\ket{\xi} = \sum_{w \in \mathcal{W}} \beta_w \ket{w}\ket{\psi_E^w} \in \H_W \kron \H_E$ now follows by linearity, and by noting that $\bra{v,\psi_E} \rho_{WE}\ket{v^\prime,\psi_E^\prime}=0=\bra{v,\psi_E} \rho^\mathrm{mix}_{WE}\ket{v^\prime,\psi_E^\prime}$ for all {\em distinct} $v,v' \in \mathcal{W}$, so that all ``cross-products'' vanish.




\end{proof}

\section{The Tightness of \refthm{quanterr}}
\label{app:tight}
We show here that in general the inequality from \refthm{quanterr} is tight. 
Specifically, we specify a natural class of sampling strategies for which \refthm{quanterr} is an equality. Informally, this class consists of sampling strategies that behave in exactly the same way if the randomized choices $T$ and $S$ are replaced by {\em fixed} choices $t_\circ$ and $s_\circ$, and instead the coordinates of $\vec{q}$ are shuffled by means of a uniformly random permutation (chosen from a subgroup of all permutations). 
The formal definition is given below, but let us point out already here that \refex{withoutreplacement} as well as \ifthenelse{\boolean{includeFifthEx}}{the QKD sampling strategy discussed in \refex{qkdsestrat}}{the sampling strategy described in Appendix~\ref{app:qkdse} (which will play a central role in the proof of quantum key distribution)} belong to this class. 
Indeed, for \refex{withoutreplacement}, instead of choosing a random subset $T$ of size $k$ one can equivalently choose a fixed subset and randomly permute the positions of $\vec{q}$. And, similarly for \ifthenelse{\boolean{includeFifthEx}}{\refex{qkdsestrat}}{the sampling strategy described in Appendix~\ref{app:qkdse}}, instead of choosing left or right from each pair $(q_{i0},q_{i1})$ at random and then choosing a random subset of size $k$ of the selected  $q_{ij}$'s, one can equivalently fix these choices and swap each pair $(q_{i0},q_{i1})$ with probability $\frac12$ and apply a random permutation to the first index. 

Let $S_n$ denote the symmetric group of degree $n$, i.e. the group of permutations on $\setn$. 
For any $\pi \in S_n$ and $\vec{q} = (q_1,\ldots,q_n) \in \mathcal{A}^n$, we write $\pi \vec{q}$ to express that $\pi$ permutes the \emph{positions} of the elements of \vec{q}, i.e., $\pi \vec{q} = (q_{\pi^{-1}(1)},\ldots,q_{\pi^{-1}(n)})$. If $\mathcal{V}$ is a set of strings $\vec{q}\in \mathcal{A}^n$, then $\pi \mathcal{V}$ means that the permutation $\pi$ acts element-wise on $\mathcal{V}$. 

\begin{definition}[$G$-Symmetry of a sampling strategy]
\label{def:gsymm}
Let $\Psi$ be a sampling strategy, let $G$ be a subgroup of $S_n$, where $n$ is the size of the population to which $\Psi$ is applied, and let $\Pi$ be a random permutation, uniformly distributed over $G$. We call $\Psi$ \emph{$G$-symmetric}, if there exist $t_\circ \subset \setn$ and $s_{\circ} \in \mathcal{S}$ such that
\[
\big(\relwt{\vec{ q}_{\bar{T}}},f(T,\vec{q}_{T},S)\big) \sim \big(\relwt{(\Pi\vec{ q})_{\bar{t}_\circ}},f(t_\circ,(\Pi \vec{q})_{t_\circ},s_\circ)\big) 
\]
where ``$\sim$'' means that the pairs have the same probability distribution.
\end{definition}
\noindent
A direct consequence of this definition is the following relation, which we will apply later in this section.
\begin{align*}
B_{T,S}^\delta &=\Set{\vec{q}\in \set{0,1}^n}{ \left| \relwt{\vec{q}\idx{\bar{T}}} -  f(T,\vec{q}\idx{T},S) \right| < \delta } 
\\ &\sim \Set{\vec{q}\in \set{0,1}^n}{ \left| \relwt{(\Pi\vec{q})\idx{\bar{t}_\circ}} -  f(t_\circ,(\Pi \vec{q})\idx{t_\circ},s_\circ) \right| < \delta } = \Pi^{-1}B_{t_\circ,s_\circ}^{\delta} .
\end{align*}
\noindent We can now rephrase \refprop{natural} and prove it.
\setcounter{prop}{0}
\begin{prop}[Rephrased]
For any \emph{$G$-symmetric} sampling strategy $\Psi^\text{sym}_{G}$ and any $\delta>0$:
\[
\qep(\Psi^\text{sym}_{G})= \sqrt{\cep(\Psi^\text{sym}_{G})}
\]
\end{prop}
\begin{proof}
We need to show that there exists a system $E$ and a state $\ket{\varphi_{AE}}$ such that $\dist\bigl(\rho_{TSAE},\tilde{\rho}_{TSAE}\bigr)^2 = \cep$ for $\tilde{\rho}_{TSAE}$ that minimizes the left hand side. 
As pointed out after the proof of \refthm{quanterr}, the particular construction of $\tilde{\rho}_{TSAE}$ used in the proof of \refthm{quanterr} does minimize $\dist\bigl(\rho_{TSAE},\tilde{\rho}_{TSAE}\bigr)$. Hence, it suffices to show that there exists a system $E$ and a state $\ket{\varphi_{AE}}$ (that depends on $G$) such that
\phantomsection\label{linktarget}
\begin{align*}
\dist\bigl(\rho_{TSAE},\tilde{\rho}_{TSAE}\bigr)^2 \stackrel{\addtocounter{equation}{1}\text{\scriptsize(\theequation)}}{=}  \left[\sum_{t,s} P_{TS}(t,s)|\inprod{\varphi_{AE}}{\tilde{\varphi}^{ts\perp}_{AE}}|\right]^2 
\stackrel{\addtocounter{equation}{1}\text{\scriptsize (\theequation)}}{=}  \sum_{t,s} P_{TS}(t,s)|\inprod{\varphi_{AE}}{\tilde{\varphi}^{ts\perp}_{AE}}|^2
\stackrel{\addtocounter{equation}{1}\text{\scriptsize (\theequation)}}{=}  \cep.
\end{align*}
where $\tilde{\rho}_{TSAE}$ and $\ket{\tilde{\varphi}^{ts\perp}_{AE}}$ are constructed as in the proof of \refthm{quanterr}. 
The derivation of equality (\addtocounter{equation}{-2}\hyperref[linktarget]{\theequation}) can be found in the proof of \refthm{quanterr}. The outline of the remaining part of the proof is as follows; we first present a candidate for \ket{\varphi_{AE}} and then we show that equalities (\addtocounter{equation}{1}\hyperref[linktarget]{\theequation}) and (\addtocounter{equation}{1}\hyperref[linktarget]{\theequation}) do indeed hold for this state. 

We choose $E$ to be empty. Furthermore, we
define
\begin{equation*}
\label{eq:symm}
\ket{\varphi_{AE}}:=\frac{1}{\sqrt{|G|}}\sum_{\pi \in G} \ket{\pi\vec{q}^*}. 
\end{equation*}
where $\vec{q}^* $ is such that $\Pr[\vec{q}^* \notin \Brv]  = \cep$. 
It follows from the projection construction for $\tilde{\rho}_{TSAE}$ that
\[
\ket{\tilde{\varphi}_{AE}^{ts\perp} }= \frac{1}{\sqrt{|H_{t,s}|}}\sum_{\pi \in H_{t,s}} \ket{\pi\vec{q}^*},
\]
where $H_{t,s}\subseteq G$, i.e. $
H_{t,s} :=  \set{\pi \in G: \pi\vec{q}^*\notin \B }$.  

To prove equality (\addtocounter{equation}{-1}\hyperref[linktarget]{\theequation}), we need to show that the inner product $|\inprod{\varphi_{AE}}{\tilde{\varphi}^{ts\perp}_{AE}}|$ is independent of $t$ and $s$. Because \ket{\varphi_{AE}} is a uniform superposition over permutations of $\vec{q}^*$ and \ket{\tilde{\varphi}^{ts\perp}_{AE}} is a renormalized projection of  \ket{\varphi_{AE}}, we can easily compute this inner product, $|\inprod{\varphi_{AE}}{\tilde{\varphi}^{ts\perp}_{AE}}| =|H_{t,s}| / \sqrt{|G|\cdot |H_{t,s}|} =  \sqrt{|H_{t,s}|/|G|}$. It suffices to show that $|H_{t,s}|$ is independent of $(t,s)$. It follows from the $G$-symmetry that there exists a $\pi$ such that $\B = \pi B^\delta_{t_\circ,s_\circ}$. Furthermore, let $\Pi$ be a random permutation, uniformly distributed over $G$. By definition of $H_{t,s}$ and because $\Pi$ is \emph{uniformly} distributed over $G$, we may write 
\addtocounter{equation}{1}
\begin{equation}
|H_{t,s}|=|G|\cdot \Pr [\Pi \,\vec{q}^* \notin \B] = |G|\cdot \Pr [\vec{q}^* \notin \Pi^{-1} \pi B^\delta_{t_\circ,s_\circ}] = |G|\cdot \Pr [\vec{q}^* \notin \Pi^{-1} B^\delta_{t_\circ,s_\circ}] ,
\label{eq:prf1}
\end{equation}
where the last expression is clearly independent of $(t,s)$.

Now, let us focus on equality (\addtocounter{equation}{-1}\hyperref[linktarget]{\theequation}). We derived in the proof of \refthm{quanterr} that 
$
\sum_{t,s} P_{TS}(t,s) \, |\inprod{\varphi_{AE}}{\tilde{\varphi}^{ts\perp}_{AE}}|^2 = \sum_{\vec{q}} P_{\vec{Q}}(\vec{q})\, \Pr\bigl[\vec{q} \!\notin\! B_{T,S}^{\delta}\bigr]
$, 
where the random variable $\vec{Q}$ is obtained by measuring subsystem $A$ of \ket{\varphi_{AE}}. By definition of $\ket{\varphi_{AE}}$, $P_{\vec{Q}}(\vec{q})>0$ only for $\vec{q}$ of the form $\pi \vec{q}^*$ for some $\pi \in G$. Hence, to prove equality (\hyperref[linktarget]{\theequation}\addtocounter{equation}{1}), we have to show that for any $\pi \in G$, $\Pr[\pi \vec{q}^* \notin \Brv]  = \cep$. This follows directly from the $G$-symmetry,
\begin{equation}
\Pr[\pi \vec{q}^*\!\notin\!\Brv] = \Pr[\pi\vec{q}^*\!\notin\!\Pi^{-1} B_{t_\circ,s_\circ}^\delta] = \Pr[\vec{q}^*\!\notin\! \pi^{-1} \Pi^{-1} B_{t_\circ,s_\circ}^\delta]= \Pr[\vec{q}^*\!\notin\!\Pi^{-1}B_{t_\circ,s_\circ}^\delta ] = \Pr[\vec{q}^*\!\notin\!\Brv ].
\label{eq:prf2}
\end{equation}
Finally, note that \refeq{prf1} and \refeq{prf2} rely on the group structure of $G$.
\end{proof}

\end{appendix}

\end{document}